\documentclass[10pt, a4paper,american]{article}

\usepackage[pdftex]{graphicx}

\usepackage{amsmath, amssymb, amsthm, amsfonts}

\usepackage{comment}

\newcommand{\argmax}{\mathop{\rm arg~max}\limits}

\global\long\def\G{\mathcal{G}}%
\global\long\def\S{\mathcal{S}}%
\global\long\def\R{\mathbb{R}}%
\global\long\def\C{\mathbb{C}}%
\global\long\def\N{\mathbb{N}}%
\global\long\def\M{\mathcal{M}}%
\global\long\def\X{\mathcal{X}}%
\global\long\def\F{\mathcal{F}}%
\global\long\def\B{\mathcal{B}}%
\global\long\def\H{\mathcal{H}}%
\global\long\def\Y{\mathcal{Y}}%
\global\long\def\A{\mathcal{A}}%
\global\long\def\bra#1{\left\langle #1\right|}%
\global\long\def\ket#1{\left|#1\right\rangle }%
\global\long\def\Tr{{\rm Tr}\,}%
\global\long\def\braket#1#2{\left\langle #1\mid#2\right\rangle }%
\global\long\def\re{{\rm Re}\,}%
\global\long\def\id{{\rm id}}%
\global\long\def\ii{\sqrt{-1}}%

\newtheorem{thm}{\protect\theoremname}
\newtheorem{lem}[thm]{\protect\lemmaname}

\providecommand{\lemmaname}{Lemma}
\providecommand{\theoremname}{Theorem}

\begin{document}

\title{Gill and Massar type bound for estimation of $SU(2)$ channel}
\author{
	Koichi Yamagata%
	\thanks{yamagata@se.kanazawa-u.ac.jp}\\
	{Institute of Science and Engineering, Kanazawa University} \\
	{Kanazawa, Ishikawa, 920-1192, Japan}
}%
\date{}

\maketitle

\begin{abstract}
In the estimation for a parametric family of quantum state on a Hilbert
space $\H$, the Gill and Massar bound is known as a lower bound of
weighted traces of covariances of unbiased estimators. The Gill and
Massar bound is derived by considering the convexity of the set of
classical Fisher information matrices, and the bound is locally achievable
by using randomized strategies when $\H=\C^{2}$. In this paper, we
show that estimation for a parametric $SU(2)$ unitary channel model
has a similar convex structure as qubit state model, and a Gill and
Massar type lower bound of weighted traces of covariances of unbiased
estimators can be derived for any weight matrix. We show that the
Gill and Massar type lower bound is achievable by using randomized
strategies when certain conditions are satisfied. To derive a convex
structure of the set of classical Fisher information matrices, we
introduce a Fisher information matrix $J^{(U)}$ for a $SU(2)$ unitary
channel model, and we show a upper bound of inverse $J^{(U)}$ weighted
trace of classical Fisher information matrix. The optimal randomized
strategy we construct in this paper does not require ancilla systems
in many cases. 

\end{abstract}

%
%
%
%
%

\section{Introduction}

Quantum state estimation and quantum channel estimation are important
topics in quantum information processing. When dealing with multi-parameter
models, the estimation problem is complicated due to the non-commutativity
in quantum mechanics\cite{rev1,rev2,rev3,rev4}. This paper deals
with the $SU(2)$ unitary channel estimation problem based on quantum
state estimation theory, so let us begin with an introduction to the state
estimation.

Let $\S(\H)$ be a set of density operators on a Hilbert space $\H$.
In quantum state estimation theory of a parametric family of density
operators $\left\{ \rho_{\theta}\in\S(\H)\mid\,\theta\in\Theta\subset\R^{d}\right\} $,
an estimator is a pair $(M,\hat{\theta})$ of a POVM $M$ taking values
in a finite set $\X$ and a map $\hat{\theta}:\X\mapsto\Theta$. An
estimator $(M,\hat{\theta})$ is called \textbf{unbiased} if
\begin{equation}
E_{\theta}[M,\hat{\theta}]=\sum_{x\in\X}\hat{\theta}(x)\Tr\rho_{\theta}M_{x}=\theta\label{eq:unbias}
\end{equation}
is satisfied for all $\theta\in\Theta$. An estimator $(M,\hat{\theta})$
is called \textbf{locally unbiased} \cite{holevo} at a fixed point
$\theta_{0}\in\Theta$ if the condition (\ref{eq:unbias}) is satisfied
around $\theta_{0}$ up to the first order of the Taylor expansion,
i.e.,
\begin{align*}
\sum_{x\in\X}\hat{\theta}^{i}(x)\Tr\rho_{\theta_{0}}M_{x} & =\theta_{0}^{i}\qquad(i=1,\dots,d),\\
\sum_{x\in\X}\hat{\theta}^{i}(x)\Tr\partial_{j}\rho_{\theta_{0}}M_{x} & =\delta_{j}^{i}\qquad(i,j=1,\dots,d),
\end{align*}
where $\partial_{j}\rho_{\theta_{0}}=\left.\frac{\partial}{\partial\theta^{j}}\rho_{\theta}\right|_{\theta=\theta_{0}}$.
To find optimal estimator, lower bounds of the weighted trace of the
covariance matrix $\Tr WV_{\theta_{0}}[M,\hat{\theta}]$ for a locally
unbiased estimator $(M,\hat{\theta})$ was investigated, where $W$
is a $d\times d$ positive real matrix and $V[M,\hat{\theta}]$ is
the covariance matrices. Because of the classical Cram\'er-Rao inequality
$V_{\theta_{0}}[M,\hat{\theta}]\geq J_{\theta_{0}}^{(C,M)^{-1}}$
with the classical Fisher information matrix 
\[
J_{\theta_{0},ij}^{(C,M)}=\sum_{x\in\X}\frac{\left(\Tr\partial_{i}\rho_{\theta_{0}}M_{x}\right)\left(\Tr\partial_{j}\rho_{\theta_{0}}M_{x}\right)}{\Tr\rho_{\theta_{0}}M_{x}},
\]
and its sharpness, $\inf_{(M,\hat{\theta})}\Tr WV_{\theta_{0}}[M,\hat{\theta}]$
is identical to $\inf_{M}\Tr WJ_{\theta_{0}}^{(C,M)^{-1}}$. 

When $\dim\H$=2, it is known that $\Tr WV[M,\hat{\theta}]$ satisfies
the following inequality:
\[
\Tr WV_{\theta_{0}}[M,\hat{\theta}]\geq c_{\theta_{0},W}^{(GM)}
\]
with the Gill and Massar bound
\begin{equation}
c_{\theta_{0},W}^{(GM)}:=\left(\Tr\sqrt{J_{\theta_{0}}^{(S)^{-1/2}}WJ_{\theta_{0}}^{(S)^{-1/2}}}\right)^{2},\label{eq:gill_massar_intr}
\end{equation}
where $J_{\theta_{0}}^{(S)}$ is the symmetric logarithmic derivative
(SLD) Fisher information matrix \cite{gill_massar,qubit_hayashi}
(see also \cite{yama_tomo}). The lower bound $c_{\theta_{0},W}^{(GM)}$
is achievable by using a randomized strategy of $d$ projective measurements
for $d=2,3$. The bound $c_{\theta_{0},W}^{(GM)}$ can be derived
by using the convex structure of the set $\F_{\theta_{0}}:=\left\{ J_{\theta_{0}}^{(C,M)}\mid M\text{ is a POVM}\right\} $
of classical Fisher information matrices. In the estimation of a qubit
state model $\left\{ \rho_{\theta}\in\S(\C^{2})\mid\theta\in\Theta\subset\R^{3}\right\} $,
the set $\F_{\theta_{0}}$ of classical Fisher information matrices
satisfies
\[
\F_{\theta_{0}}=\left\{ F\mid F\text{ is a \ensuremath{d\times d} real positive matrix such that }\Tr J^{(S)^{-1}}F\leq1\right\} .
\]
Thus finding the minimum of weighted trace of covariance matrices
for a weight $W$ is reduced to the convex optimization problem to
find
\[
\min_{F\in\F_{\theta_{0}}}\Tr WF^{-1}(=c_{\theta_{0},W}^{(GM)}),
\]
and the minimum is achieved by a randomized strategy. In this paper,
we show that the estimation of $SU(2)^{\otimes n}$ channel has a
similar structure and a similar bound can be derived. 

Let $\left\{ U_{\theta}\mid\,\theta\in\Theta\subset\R^{d}\right\} $
be a parametric family of $SU(2)$ unitary matrices on a Hilbert space
$\H=\C^{2}$ with $d=1,2,3$, and let us consider a corresponding
parametric family of unitary channels 
\[
\left\{ \Gamma_{\theta}:\B(\C^{2})\to\B(\C^{2}),\,\rho\mapsto U_{\theta}\rho U_{\theta}^{*}\mid\,\theta\in\Theta\subset\R^{d}\right\} 
\]
 on the Hilbert space $\C^{2}$, where $\B(\C^{2})$ is a set of linear
operators on $\C^{2}$. In this study, we deal with the $n$th i.i.d.
extended model
\begin{equation}
\G^{(n)}:=\left\{ \Gamma_{\theta}^{\otimes n}:\rho\mapsto U_{\theta}^{\otimes n}\rho U_{\theta}^{*\otimes n}\mid\theta\in\Theta\subset\R^{d}\right\} .\label{eq:su2n_model}
\end{equation}

To estimate an unknown parameter $\theta\in\Theta$, an ancilla Hilbert
space $\H_{a}^{(n)}$ may improve the estimation accuracy, and an
estimator is a triple $(\rho^{(n)},M^{(n)},\hat{\theta}^{(n)})$ of
an input quantum state $\rho^{(n)}\in\S(\C^{2^{n}}\otimes\H_{a}^{(n)})$,
a POVM $M^{(n)}$ on $\C^{2^{n}}\otimes\H_{a}^{(n)}$ taking values
in a finite set $\X^{(n)}$, and a map $\hat{\theta}^{(n)}:\X^{(n)}\mapsto\Theta$.
An estimator $(\rho^{(n)},M^{(n)},\hat{\theta}^{(n)})$ is called
locally unbiased at $\theta_{0}\in\Theta$ if
\begin{align*}
\sum_{x\in\X^{(n)}}\hat{\theta}^{(n)i}(x)\Tr(\Gamma_{\theta_{0}}^{\otimes n}\otimes\id_{\H_{a}^{(n)}})(\rho^{(n)})M_{x}^{(n)} & =\theta_{0}^{i}\qquad(i=1,\dots,d),\\
\sum_{x\in\X^{(n)}}\hat{\theta}^{(n)i}(x)\Tr\partial_{j}(\Gamma_{\theta_{0}}^{\otimes n}\otimes\id_{\H_{a}^{(n)}})(\rho^{(n)})M_{x}^{(n)} & =\delta_{j}^{i}\qquad(i,j=1,\dots,d),
\end{align*}
is satisfied, where $\id_{\H_{a}^{(n)}}$ is an identity channel on
$\H_{a}^{(n)}$. 

There are several existing studies on the estimation of unitary channels
\cite{fujiwara_dense,imai_sud,hayashi_su2,Kahn_sud}. In their studies,
the symmetric structure of the unitary group was focused, and three
symmetric parameters $\theta^{1},\theta^{2}$ and $\theta^{3}$ were
treated as equal, and maximally entangled states were used as input.
However, if the estimation accuracy of $\theta^{1}$ is more important
than $\theta^{2}$ and $\theta^{3}$ for an experimenter, the existing
estimation methods may not be optimal. In this study, we consider
the lower bound of the weighted trace of covariance matrix $\Tr WV_{\theta_{0}}^{(n)}[\rho^{(n)},M^{(n)},\hat{\theta}^{(n)}]$
for a locally unbiased estimator $(\rho^{(n)},M^{(n)},\hat{\theta}^{(n)})$
of $\left\{ \Gamma_{\theta}^{\otimes n}\right\} _{\theta\in\Theta}$
with an arbitrary $d\times d$ positive real matrix $W$ as a weight.
We show that maximally entangled states are not optimal in general,
and randomized strategies without using an ancilla Hilbert space are
optimal in many cases. 

In this paper, we define a Fisher information matrix at $\theta_{0}\in\Theta$
for the $SU(2)$ channel model by
\begin{equation}
J_{\theta_{0}}^{(U)}:=2\left[\Tr\left(\partial_{i}U_{\theta_{0}}\right)^{*}\left(\partial_{j}U_{\theta_{0}}\right)\right]_{1\leq i,j\leq d}\label{eq:ch-Fisher}
\end{equation}
for $\theta_{0}\in\Theta$, and we show a matrix inequality
\begin{equation}
n^{2}V_{\theta_{0}}^{(n)}[\rho^{(n)},M^{(n)},\hat{\theta}^{(n)}]\geq J_{\theta_{0}}^{(U)^{-1}}.\label{eq:matrix_inq_intr}
\end{equation}
Note that $J_{\theta_{0}}^{(U)}$ is a real positive matrix. When
$n=1$, the equality is achievable by a maximally entangled input
state. When $n=2$ and $d=2$, the equality is achievable by a pure
state input given in Theorem \ref{thm:bound_matrix} without an ancilla
Hilbert space. When $(d-1)n>2$, the matrix inequality (\ref{eq:matrix_inq_intr})
is not sharp. 

Similar to the estimation of a qubit state model that has Gill and
Massar bound, there exists a similar convex structure in the estimation
of $SU(2)$ channel model $\G^{(n)}$. We prove that the set $\F_{\theta_{0}}^{(n)}$
of classical Fisher information matrices at $\theta_{0}$ for $\G^{(n)}$
satisfies
\begin{equation}
\F_{\theta_{0}}^{(n)}\subset\left\{ F\mid F\text{ is a \ensuremath{d\times d} real positive matrix such that }\Tr J_{\theta_{0}}^{(U)^{-1}}F\leq n^{2}+2n\right\} .\label{eq:gill_massar_int0}
\end{equation}
By using this, we prove
\begin{equation}
(n^{2}+2n)\Tr WV_{\theta_{0}}^{(n)}[\rho^{(n)},M^{(n)},\hat{\theta}^{(n)}]\geq c_{\theta_{0},W}\label{eq:gill_massar_int}
\end{equation}
with a lower bound
\[
c_{\theta_{0},W}:=\left(\Tr\sqrt{J_{\theta_{0}}^{(U)^{-1/2}}WJ_{\theta_{0}}^{(U)^{-1/2}}}\right)^{2},
\]
for a real positive $d\times d$ matrix $W$. We can see that this
lower bound $c_{\theta_{0},W}$ is similar to the Gill and Massar
bound (\ref{eq:gill_massar_intr}) for the qubit state estimation.
When $d=3$ and $n\geq\max\left\{ 3,\frac{\sqrt{w_{2}}+\sqrt{w_{3}}}{\sqrt{w_{1}}}-1\right\} $,
the lower bound $c_{\theta_{0},W}$ is achievable by using a randomized
strategy without an ancilla Hilbert space, where $w_{1},w_{2},w_{3}$
are eigenvalues of $\tilde{W}:=J_{\theta_{0}}^{(U)^{-1/2}}WJ_{\theta_{0}}^{(U)^{-1/2}}$
such that $0\leq w_{1}\leq w_{2}\leq w_{3}$. When $d=3$ and $n=2$,
more informative and achievable bound can be obtained by considering
\[
\F_{\theta_{0}}^{(n)}\subset\left\{ F\mid F\text{ is a \ensuremath{d\times d} real positive matrix such that }\Tr J_{\theta_{0}}^{(U)^{-1}}F\leq n^{2}+2n\text{ and }n^{2}F\leq J_{\theta_{0}}^{(U)}\right\} .
\]
When $d=2$, $n^{2}+2n$ in (\ref{eq:gill_massar_int0}) and (\ref{eq:gill_massar_int})
can be replaced by $2\left\lfloor (n^{2}+2n)/2\right\rfloor $, and
the lower bound $c_{\theta_{0},W}$ is achievable if $W=J_{\theta_{0}}^{(U)}$.
For an arbitrary weight $W$ and $d=2$, the lower bound $c_{\theta_{0},W}$
is asymptotically achievable by using randomized strategies. 

Although our theory assume locally unbiased estimators, in many cases
the local theory of estimation can be used for global estimation by
using adaptive estimation methods\cite{adaptive}. On the other hand,
the above results show that the optimal covariance have the so-called
Heisenberg scaling $1/n^{2}$. It was reported that if a model has
the Heisenberg limit, the variance with respect to the total number
of samples may not reach the local limit \cite{rev1,rev5}. It is
open problem whether the local lower bound (\ref{eq:gill_massar_int})
can be achieved asymptotically if $n$ is the total number of samples.
At least, if the estimation of $U_{\theta}^{\otimes n}$ for a fixed
$n$ is repeated a sufficiently large number of times $M$, the local
limit (\ref{eq:gill_massar_int}) can work globally. 

This paper is organized as follows: In Section \ref{sec:mix_input},
we introduce some basic facts that commonly hold in the estimation
of quantum channel models. We show here that a set of classical Fisher
information matrices is convex by considering randomized strategies.
Further, we show that any estimator for a channel estimation can be
reproduced by using a pure state as input with an ancilla Hilbert
space whose dimensions is less than $\dim\H_{in}$ due to a generalized
purification. In Section \ref{sec:pure_est}, we review the estimation
theory for pure state motels. Section \ref{sec:est_su2} is devoted
to the estimation of $SU(2)$ channel models. In Subsection \ref{subsec:su2_model},
we prepare $SU(2)$ channel models and notations. In Subsection \ref{subsec:matrix_inequality},
we derive the matrix inequality (\ref{eq:matrix_inq_intr}). In Subsection
\ref{subsec:tr_inquality}, we show upper bounds of traces of classical
Fisher information matrices to derive Gill and Massar type inequalities.
In Subsection \ref{subsec:3para} and \ref{subsec:2para}, we derive
Gill and Massar type lower bounds for three parameter models and two
parameter models with optimal randomized strategies. In Subsection
\ref{subsec:estN2}, we derive Gill and Massar type bounds for $d=3$
and $n=2$. In Section \ref{sec:plots}, we show some figures to visualize
the sets of classical Fisher information matrices. Section \ref{sec:Conclusion}
is the conclusion. For the reader's convenience, some additional materials
are presented in the Appendix. In Appendix \ref{sec:purification},
we show a generalized purification. In Appendix \ref{sec:SLD_Holevo},
we review a relation between the Holevo bound and the SLD bound for
a quantum state estimation. In Appendix \ref{sec:pure_holevo}, we
review the Holevo bound for a pure state estimation. In Appendix \ref{sec:numerical_calculations},
we detail numerical calculations used to draw figures.

\section{Estimation of quantum channel and its convex structure\label{sec:mix_input}}

In this section, we introduce some basic facts that commonly hold
in the estimation of a quantum channel model. Let $\left\{ \Gamma_{\theta}:\B(\H_{in})\to\B(\H_{out})\mid\,\theta\in\Theta\subset\R^{d}\right\} $
be a parametric family of quantum channels with Hilbert spaces $\H_{in}$
and $\H_{out}$. To estimate unknown $\theta$, an ancilla Hilbert
space $\H_{a}$ may improve the estimation accuracy, and an estimator
is a triple $(\rho,M,\hat{\theta})$ of an input quantum state $\rho\in\S(\H_{in}\otimes\H_{a})$,
a POVM $M\in\M(\H_{out}\otimes\H_{a},\X)$ on $\H_{out}\otimes\H_{a}$
taking values in a finite set $\X$, and a map $\hat{\theta}:\X\mapsto\Theta$.
An estimator $(\rho,M,\hat{\theta})$ is called unbiased if
\begin{equation}
E_{\theta}[\rho,M,\hat{\theta}]=\sum_{x\in\X}\hat{\theta}(x)\Tr(\Gamma_{\theta}\otimes\id_{\H_{a}})(\rho)M(x)=\theta\label{eq:unbias2}
\end{equation}
is satisfied for all $\theta\in\Theta$, where $\id_{\H_{a}}$ is
an identity channel on $\H_{a}$. An estimator $(\rho,M,\hat{\theta})$
is called locally unbiased at a fixed point $\theta_{0}\in\Theta$
if the condition (\ref{eq:unbias2}) is satisfied around $\theta_{0}$
up to the first order of the Taylor expansion, i.e.,

\begin{align*}
E_{\theta}[\rho,M,\hat{\theta}^{i}] & =\theta_{0}^{i}\qquad(i=1,\dots,d),\\
\partial_{j}E_{\theta}[\rho,M,\hat{\theta}^{i}] & =\delta_{j}^{i}\qquad(i,j=1,\dots,d).
\end{align*}
Note that for fixed $\rho$ and $M$, 
\[
\left\{ \left.p_{\theta}^{(\rho,M)}(x):=\Tr(\Gamma_{\theta}\otimes\id_{\H_{a}})(\rho)M(x)\right|\,\theta\in\Theta\subset\R^{d}\right\} 
\]
is a parametric family of classical probability distributions, and
the covariance matrix $V_{\theta_{0}}[\rho,M,\hat{\theta}]$ of a
locally unbiased estimator $(\rho,M,\hat{\theta})$ at $\theta=\theta_{0}$
satisfies the classical Cram\'er-Rao inequality
\begin{equation}
V_{\theta_{0}}[\rho,M,\hat{\theta}]\geq J_{\theta_{0}}^{(C,\rho,M)^{-1}},\label{eq:classical_cramer_u}
\end{equation}
where
\[
J_{\theta_{0}}^{(C,\rho,M)}:=\left[\sum_{x\in\X}\frac{\left\{ \partial_{i}\Tr(\Gamma_{\theta_{0}}\otimes\id_{\H_{a}})(\rho)M(x)\right\} \left\{ \partial_{j}\Tr(\Gamma_{\theta_{0}}\otimes\id_{\H_{a}})(\rho)M(x)\right\} }{\Tr(\Gamma_{\theta_{0}}\otimes\id_{\H_{a}})(\rho)M(x)}\right]_{1\leq i,j\leq d}
\]
is the classical Fisher information matrix at $\theta=\theta_{0}$
with restrict to $\rho$ and $M$.  The equality of (\ref{eq:classical_cramer_u})
is achieved when $\hat{\theta}^{i}(x)=\theta_{0}^{i}+\sum_{j=1}^{d}\left[J_{\theta_{0}}^{(C,\rho,M)^{-1}}\right]^{ij}\frac{\partial_{j}\Tr(\Gamma_{\theta_{0}}\otimes\id_{a})(\rho)M(x)}{\Tr(\Gamma_{\theta_{0}}\otimes\id_{a})(\rho)M(x)}.$
Thus seeking an estimator $(\rho,M,\hat{\theta})$ that minimizes
$\Tr WV[\rho,M,\hat{\theta}]$ is reduced to seeking a pair $(\rho,M)$
that minimizes $\Tr WJ_{\theta_{0}}^{(C,\rho,M)^{-1}}$ for a given
$d\times d$ real positive matrix $W$. We call a state-measurement
pair $(\rho,M)$ a channel measurement. When $\rho=\ket{\psi}\bra{\psi}$
is a pure state, we also denote a channel measurement as $(\psi,M)$.
We show that it is sufficient to restrict to $\dim\H_{a}\leq\dim\H_{in}$
and $\rho\in\S(\H_{in}\otimes\H_{a})$ to a pure state as follows.

\begin{thm}
\label{thm:ch_eqiv}Let $\left\{ \Gamma_{\theta}:\B(\H_{in})\to\B(\H_{out})\mid\,\theta\in\Theta\subset\R^{d}\right\} $
be a parametric family of quantum channels. For any pair $(\rho,M)$
of a quantum state $\rho\in\S(\H_{in}\otimes\H_{a})$ and a POVM $M\in\M(\H_{out}\otimes\H_{a},\X)$
taking values in a finite set $\X$ on a Hilbert space $\H_{out}\otimes\H_{a}$with
an ancilla Hilbert space $\H_{a}$, there exists a pair $(\ket{\psi}\bra{\psi},N)$
of a pure state $\ket{\psi}\bra{\psi}\in\S(\H_{in}\otimes\H_{b})$
and a POVM $N\in\M(\H_{out}\otimes\H_{b},\X)$ on $\H_{out}\otimes\H_{b}$
with another Hilbert space $\H_{b}$ such that
\[
\dim\H_{b}\leq\dim\H_{in}
\]
and
\[
p_{\theta}^{(\rho,M)}(x)=p_{\theta}^{(\psi,N)}(x)
\]
for any $\theta\in\Theta$ and any $x\in\X$, where
\begin{align*}
p_{\theta}^{(\psi,N)}(x) & =\Tr\left(\Gamma_{\theta}\otimes\id_{\H_{b}}\right)(\ket{\psi}\bra{\psi})N(x).
\end{align*}
\end{thm}

\begin{proof}
Because of a generalized purification given in Lemma \ref{lem:to_pure},
there exists a Hilbert space $\H_{b}$ and a pure state $\ket{\psi}\bra{\psi}\in\S(\H_{in}\otimes\H_{b})$
on $\H_{in}\otimes\H_{b}$ and a quantum channel $\Lambda:\B(\H_{b})\to\B(\H_{a})$
such that $\id_{\H_{in}}\otimes\Lambda(\ket{\psi}\bra{\psi})=\rho$
and $\dim\H_{b}\leq\dim\H_{in}$. It can be proved that a POVM $N\in\M(\H_{out}\otimes\H_{b},\X)$
defined by 
\[
N(x)=\left(\id_{\H_{out}}\otimes\Lambda\right)^{*}(M(x))
\]
satisfies
\begin{align*}
p_{\theta}^{(\psi,N)}(x) & =\Tr\left(\Gamma_{\theta}\otimes\id_{\H_{b}}\right)\left(\ket{\psi}\bra{\psi}\right)N(x)\\
 & =\Tr\left(\Gamma_{\theta}\otimes\id_{\H_{b}}\right)\left(\ket{\psi}\bra{\psi}\right)\left(\id_{\H_{out}}\otimes\Lambda\right)^{*}(M(x))\\
 & =\Tr\left[\left(\id_{\H_{out}}\otimes\Lambda\right)\circ\left(\Gamma_{\theta}\otimes\id_{\H_{b}}\right)\left(\ket{\psi}\bra{\psi}\right)\right]M(x)\\
 & =\Tr\left[\left(\Gamma_{\theta}\otimes\id_{\H_{a}}\right)\circ\left(\id_{\H_{in}}\otimes\Lambda\right)\left(\ket{\psi}\bra{\psi}\right)\right]M(x)\\
 & =\Tr\left[\left(\Gamma_{\theta}\otimes\id_{\H_{a}}\right)(\rho)\right]M(x)=p_{\theta}^{(\rho,M)}(x).
\end{align*}
\end{proof}
Next, we show that the set of classical Fisher information matrices
\[
\F_{\theta_{0}}=\left\{ \left.J_{\theta_{0}}^{(C,\psi,M)}\right|\ket{\psi}\bra{\psi}\in\S(\H_{in}\otimes\H_{a}),\,M\text{ is a POVM}\right\} 
\]
has a convex structure for a fixed point $\theta_{0}\in\Theta$ and
an ancilla Hilbert space $\H_{a}$ with $\dim\H_{a}=\dim\H_{in}$.
Let $\rho_{1},\rho_{2}\in\S(\H_{in}\otimes\H_{a})$ be quantum states
on $\H_{in}\otimes\H_{a}$, and let $M_{1}\in\M(\H_{in}\otimes\H_{a},\X)$
and $M_{2}\in\M(\H_{out}\otimes\H_{a},\Y)$ be POVMs on $\H_{out}\otimes\H_{a}$
taking values in $\X$ and $\Y$. Let us consider a randomized channel
measurement such that the channel measurement $\left(\rho_{1},M_{1}\right)$
is applied with probability $q$ and the channel measurement $\left(\rho_{2},M_{2}\right)$
is applied with probability $1-q$. We denote such a randomized channel
measurement by 
\[
q\left(\rho_{1},M_{1}\right)\oplus(1-q)\left(\rho_{2},M_{2}\right).
\]
Applying this randomized channel measurement yields a family of classical
probability distributions
\[
\left\{ \left.p_{\theta}^{\left(q\left(\rho_{1},M_{1}\right)\oplus(1-q)\left(\rho_{2},M_{2}\right)\right)}(z)\right|\,\theta\in\Theta\subset\R^{d},z\in\X\cup\Y\right\} ,
\]
where
\begin{equation}
p_{\theta}^{\left(q\left(\rho_{1},M_{1}\right)\oplus(1-q)\left(\rho_{2},M_{2}\right)\right)}(z)=\begin{cases}
qp_{\theta}^{\left(\rho_{1},M_{1}\right)} & \text{if \ensuremath{z\in\X},}\\
(1-q)p_{\theta}^{\left(\rho_{2},M_{2}\right)} & \text{if \ensuremath{z\in\Y}}.
\end{cases}\label{eq:mix_prob}
\end{equation}

\begin{thm}
\label{thm:convex_fisher}For a randomized channel measurement $q\left(\rho_{1},M_{1}\right)\oplus(1-q)\left(\rho_{2},M_{2}\right)$,
there exists a channel measurement $(\psi,N)$ with a pure state $\ket{\psi}\bra{\psi}\in\S(\H_{in}\otimes\H_{a})$
and a POVM $N\in\M(\H_{out}\otimes\H_{a},\X\cup\Y)$ on $\H_{out}\otimes\H_{a}$
such that
\begin{equation}
\Tr\left(\Gamma_{\theta}\otimes\id_{\H_{a}}\right)(\ket{\psi}\bra{\psi})N(z)=p_{\theta}^{\left(q\left(\rho_{1},M_{1}\right)\oplus(1-q)\left(\rho_{2},M_{2}\right)\right)}(z),\label{eq:mix2single}
\end{equation}
for any $\theta\in\Theta$ and any $z\in\X\cup\Y$. Further, the classical
Fisher information matrix $J_{\theta}^{\left(C,q\left(\rho_{1},M_{1}\right)\oplus(1-q)\left(\rho_{2},M_{2}\right)\right)}$
of $\left\{ p_{\theta}^{\left(q\left(\rho_{1},M_{1}\right)\oplus(1-q)\left(\rho_{2},M_{2}\right)\right)}(z)\right\} _{\theta}$
satisfies
\begin{align}
J_{\theta}^{\left(C,\psi,N\right)} & =J_{\theta}^{\left(C,q\left(\rho_{1},M_{1}\right)\oplus(1-q)\left(\rho_{2},M_{2}\right)\right)}\label{eq:mixFisher1}\\
 & =qJ_{\theta}^{\left(C,\rho_{1},M_{1}\right)}+(1-q)J_{\theta}^{\left(C,\rho_{2},M_{2}\right)}\label{eq:mixFisher2}
\end{align}
for any $\theta\in\Theta$. In particular, $\F_{\theta}$ is convex
for any $\theta\in\Theta$. 
\end{thm}

\begin{proof}
Let $\left\{ \ket{e_{i}}\right\} _{i=1}^{2}$ be an orthonormal basis
of $\H_{c}=\C^{2}$. Let
\[
\tilde{\rho}:=q\rho_{1}\otimes\ket{e_{1}}\bra{e_{1}}+(1-q)\rho_{2}\otimes\ket{e_{2}}\bra{e_{2}}
\]
and
\[
\tilde{M}(z):=\begin{cases}
M_{1}(z)\otimes\ket{e_{1}}\bra{e_{1}} & \text{if \ensuremath{z\in\X},}\\
M_{2}(z)\otimes\ket{e_{2}}\bra{e_{2}} & \text{if \ensuremath{z\in\Y}},
\end{cases}
\]
be a quantum state and a POVM taking values in $\X\cup\Y$. We can
easily check the classical probability distributions satisfy
\[
p_{\theta}^{\left(q\left(\rho_{1},M_{1}\right)\oplus(1-q)\left(\rho_{2},M_{2}\right)\right)}(z)=p_{\theta}^{\left(\tilde{\rho},\tilde{M}\right)}(z)
\]
for any $z\in\X\cup\Y$ and $\theta\in\Theta$. Due to Theorem \ref{thm:ch_eqiv},
there exists a channel measurement $(\psi,N)$ with a pure state $\ket{\psi}\bra{\psi}\in\S(\H_{in}\otimes\H_{a})$
and a POVM $N\in\M(\H_{out}\otimes\H_{a},\X\cup\Y)$ such that
\begin{align*}
\Tr\left(\Gamma_{\theta}\otimes\id_{\H_{a}}\right)(\ket{\psi}\bra{\psi})N(z) & =\Tr\left(\Gamma_{\theta}\otimes\id_{\H_{a}}\otimes\id_{\H_{c}}\right)(\tilde{\rho})\tilde{M}(z)\\
 & =p_{\theta}^{\left(\tilde{\rho},\tilde{M}\right)}(z)
\end{align*}
for any $\theta\in\Theta$ and any $z\in\X\cup\Y$. This proves (\ref{eq:mix2single})
and (\ref{eq:mixFisher1}). The equation (\ref{eq:mixFisher2}) is
proved from the direct calculation of the classical Fisher information
matrix of (\ref{eq:mix_prob}). 
\end{proof}

\section{Estimation for pure state model\label{sec:pure_est}}

In this section, we review estimation theory for pure state motel.
Let $\left\{ \rho_{\theta}\mid\,\theta\in\Theta\subset\R^{d}\right\} $
be a quantum statistical model comprising pure states on a finite
dimensional Hilbert space $\H$. The observables 
\begin{equation}
L_{\theta_{0},i}:=2\partial_{i}\rho_{\theta_{0}}\label{eq:sld_pure}
\end{equation}
for $1\leq i\leq d$ satisfy
\begin{equation}
\partial_{i}\rho_{\theta_{0}}=\frac{1}{2}\left(\rho_{\theta_{0}}L_{\theta_{0},i}+L_{\theta_{0},i}\rho_{\theta_{0}}\right),\label{eq:SLD_def}
\end{equation}
and such observables are called symmetric logarithmic derivatives
(SLDs) at $\theta_{0}\in\Theta$. Note that the observables that satisfies
(\ref{eq:SLD_def}) is not unique. In fact, a observable $L_{\theta_{0},i}+K_{i}$
also satisfies (\ref{eq:SLD_def}) for any observable $K_{i}$ such
that $\rho_{\theta_{0}}K_{i}=0$. Nevertheless, the SLD Fisher information
matrix $J_{\theta_{0}}^{(S)}$ at $\theta_{0}\in\Theta$ defined by
\begin{equation}
J_{\theta_{0},ij}^{(S)}=\re\Tr\rho_{\theta_{0}}L_{\theta_{0},i}L_{\theta_{0},j}=4\re\Tr\rho_{\theta_{0}}(\partial_{i}\rho_{\theta_{0}})(\partial_{j}\rho_{\theta_{0}})\label{eq:SLDFisher}
\end{equation}
is unique since $K_{i}$ is vanished in (\ref{eq:SLDFisher}). 

Let $\ket{\psi}\in\H$ be an unit vector such that $\rho_{\theta_{0}}=\ket{\psi}\bra{\psi}$,
and let 
\[
\ket{l_{i}}=L_{\theta_{0},i}\ket{\psi}
\]
be the vector representation of the SLD $L_{\theta_{0},i}$\cite{matsumoto_pure}.
By using these vectors, we can obtain
\[
\partial_{i}\rho_{\theta_{0}}=\frac{1}{2}\left(\ket{\psi}\bra{l_{i}}+\ket{l_{i}}\bra{\psi}\right)
\]
and
\[
J_{\theta_{0},ij}^{(S)}={\rm Re}\braket{l_{i}}{l_{j}}.
\]

The inverse SLD Fisher information matrix $J_{\theta_{0}}^{(S)^{-1}}$
is known as a lower bound of the covariance matrix $V_{\theta_{0}}[M,\hat{\theta}]$
of any locally unbiased estimator $\left(M,\hat{\theta}\right)$ for
a pure state model, and the following theorem gives a necessary and
sufficient condition to achieve the SLD lower bound\cite{matsumoto_pure}. 
\begin{thm}
\label{thm:sld_ineq_pure}For any POVM $M$, the classical Fisher
information matrix $J_{\theta_{0}}^{(C,M)}$ with respect to $M$
satisfies
\begin{equation}
J_{\theta_{0}}^{(C,M)}\leq J_{\theta_{0}}^{(S)}.\label{eq:sld_ineq_pure}
\end{equation}
There exists a POVM $M$ to achieve the equality if and only if
\begin{equation}
Z_{\theta_{0}}:=\left[\Tr\rho_{\theta_{0}}L_{\theta_{0},j}L_{\theta_{0},i}\right]_{1\leq i,j\leq d}=\left[\braket{l_{i}}{l_{j}}\right]_{1\leq i,j\leq d}\label{eq:sld_ineq_pure_cond}
\end{equation}
is a real matrix.
\end{thm}

More generally, for the quantum statistical model comprising pure
states, a more informative and sharp lower bound of weighted traces
of covariances is known as the Holevo bound (\ref{eq:holevo_bound})
(see Theorem \ref{thm:pure_holevo} in Appendix \ref{sec:pure_holevo}).
It is known that when (\ref{eq:sld_ineq_pure_cond}) is a real matrix,
$\Tr WJ_{\theta_{0}}^{(S)^{-1}}$ and the Holevo bound $c_{\theta_{0},W}^{(H)}$
coincide for any weight matrix $W$ (see Lemma \ref{lem:sld_holevo}
in Appendix \ref{sec:SLD_Holevo}). Theorem \ref{thm:sld_ineq_pure}
can be proved as a corollary of Lemma \ref{lem:sld_holevo} and Theorem
\ref{thm:pure_holevo}. 

\section{Estimation of $SU(2)$ channel by using the convex structure\label{sec:est_su2}}

\subsection{$SU(2)$ channel model\label{subsec:su2_model}}

Let 
\begin{equation}
\left\{ \Gamma_{\theta}:\rho\mapsto U_{\theta}\rho U_{\theta}^{*}\mid\theta\in\Theta\subset\R^{d}\right\} 
\end{equation}
be a parametric family of $SU(2)$ channels on a Hilbert space $\H=\C^{2}$
where $U_{\theta}$ is in $SU(2)$ with $1\leq d\leq3$. In this section,
we consider the estimation of the $n$-i.i.d. extended model $\left\{ \Gamma_{\theta}^{\otimes n}\mid\theta\in\Theta\subset\R^{3}\right\} $.
Due to Theorem \ref{thm:ch_eqiv}, the input states for the the optimal
estimator can be restricted to pure states on a Hilbert space $\H^{\otimes n}\otimes\H_{a}^{(n)}$
with $\dim\H_{a}^{(n)}=2^{n}$. Due to Theorem \ref{thm:convex_fisher},
the set
\begin{equation}
\F_{\theta_{0}}^{(n)}=\left\{ \left.J_{\theta_{0}}^{(C,n,\psi,M)}\right|\,\ket{\psi}\bra{\psi}\in\S(\H^{\otimes n}\otimes\H_{a}^{(n)}),\,M\in\M(\H^{\otimes n}\otimes\H_{a}^{(n)},\X)\right\} 
\end{equation}
of classical Fisher information matrices is convex for any $\theta_{0}\in\Theta$
and any $n\in\N$, where $J_{\theta_{0}}^{(C,n,\psi,M)}$ is the classical
Fisher information matrices with restrict to an input pure state $\ket{\psi}\bra{\psi}\in\S(\H^{\otimes n}\otimes\H_{a}^{(n)})$
and a POVM $M\in\M(\H^{\otimes n}\otimes\H_{a}^{(n)},\X)$ on $\H^{\otimes n}\otimes\H_{a}^{(n)}$
with an ancilla Hilbert space $\H_{a}^{(n)}=\C^{2^{n}}$.  When a
input state $\rho\in\S(\H^{\otimes n}\otimes\H_{a}^{(n)})$ is a pure
state, the output state $\left(\Gamma_{\theta}^{\otimes n}\otimes\id_{\H_{a}^{(n)}}\right)(\rho)=(U_{\theta}^{\otimes n}\otimes I_{\H_{a}^{(n)}})\rho(U_{\theta}^{\otimes n}\otimes I_{\H_{a}^{(n)}})^{*}$
is also a pure state. Thus the estimation theory for pure states given
in Section \ref{sec:pure_est} plays an important role in the estimation
of unitary channels. 

Let us define a $d\times d$ positive matrix
\begin{equation}
J_{\theta_{0}}^{(U)}:=2\left[\Tr\left(\partial_{i}U_{\theta_{0}}\right)^{*}\left(\partial_{j}U_{\theta_{0}}\right)\right]_{1\leq i,j\leq d}
\end{equation}
for $\theta_{0}\in\Theta$. Note that $\left(\partial_{i}U_{\theta_{0}}\right)U_{\theta_{0}}^{*}$
is skew-Hermitian and $J_{\theta_{0}}^{(U)}$ is a real matrix because
$\partial_{i}\left(U_{\theta_{0}}U_{\theta_{0}}^{*}\right)=\left(\partial_{i}U_{\theta_{0}}\right)U_{\theta_{0}}^{*}+U_{\theta_{0}}\left(\partial_{i}U_{\theta_{0}}\right)^{*}=0$.
Let
\begin{equation}
X_{i}=\frac{2}{\sqrt{-1}}\sum_{j=1}^{d}K_{\theta_{0},ij}\left(\partial_{j}U_{\theta_{0}}\right)U_{\theta_{0}}^{*}\qquad(1\leq i\leq d)\label{eq:X_def}
\end{equation}
be observables, where
\begin{equation}
K_{\theta_{0}}=O\sqrt{J_{\theta_{0}}^{(U)^{-1}}}\label{eq:K_def}
\end{equation}
with any $d\times d$ real orthogonal matrix $O$. The observables
$X_{1},X_{2},X_{3}$ satisfy 
\begin{equation}
\left\{ X_{i},X_{j}\right\} =\frac{1}{2}\left(X_{i}X_{j}+X_{j}X_{i}\right)=\delta_{ij}I\label{eq:pauli_relation}
\end{equation}
like Pauli matrices, where $X_{d+1},\dots,X_{3}$ are appropriate
operators when $d<3$. It can be seen from (\ref{eq:X_def}) that
\begin{equation}
\partial_{i}U_{\theta_{0}}=\frac{\sqrt{-1}}{2}\left(\sum_{j=1}^{d}\left(K_{\theta_{0}}^{-1}\right)_{ij}X_{j}\right)U_{\theta_{0}},
\end{equation}
since
\begin{align}
\sum_{i=1}^{d}\left(K_{\theta_{0}}^{-1}\right)_{si}X_{i} & =\frac{2}{\sqrt{-1}}\sum_{i=1}^{d}\sum_{j=1}^{d}\left(K_{\theta_{0}}^{-1}\right)_{si}K_{\theta_{0},ij}\left(\partial_{j}U_{\theta_{0}}\right)U_{\theta_{0}}^{*}\nonumber \\
 & =\frac{2}{\sqrt{-1}}\sum_{j=1}^{d}\delta_{sj}\left(\partial_{j}U_{\theta_{0}}\right)U_{\theta_{0}}^{*}=\frac{2}{\sqrt{-1}}\left(\partial_{s}U_{\theta_{0}}\right)U_{\theta_{0}}^{*}.\label{eq:derivate_dU}
\end{align}
Similarly, we have
\begin{equation}
\partial_{i}\left(U_{\theta_{0}}^{\otimes n}\right)=\frac{\sqrt{-1}}{2}\left(\sum_{j=1}^{d}\left(K_{\theta_{0}}^{-1}\right)_{ij}X_{j}^{(n)}\right)U_{\theta_{0}}^{\otimes n},
\end{equation}
for $n\in\N$, where
\begin{equation}
X_{j}^{(n)}:=\sum_{k=1}^{n}I^{\otimes(k-1)}\otimes X_{j}\otimes I^{\otimes(n-k)}\quad(j=1,2,3).
\end{equation}
By fixing an input vector $\ket{\psi}\in\H^{\otimes n}\otimes\H_{a}^{(n)}$,
we have a family of pure states

\[
\left\{ \left.\left(\Gamma_{\theta}^{\otimes n}\otimes id_{\H_{a}^{(n)}}\right)\left(\ket{\psi}\bra{\psi}\right)\in\S(\H^{\otimes n}\otimes\H_{a}^{(n)})\right|\,\theta\in\Theta\subset\R^{d}\right\} ,
\]
and we have an $i$th SLD
\begin{align}
L_{\theta_{0},i}^{(S,n,\psi)} & :=2\partial_{i}\left(\Gamma_{\theta_{0}}^{\otimes n}\otimes id_{\H_{a}^{(n)}}\right)\left(\ket{\psi}\bra{\psi}\right)\\
 & =2\partial_{i}\left((U_{\theta_{0}})^{\otimes n}\otimes I_{\H_{a}^{(n)}}\ket{\psi}\bra{\psi}(U_{\theta_{0}}^{*})^{\otimes n}\otimes I_{\H_{a}^{(n)}}\right)\\
 & =\sqrt{-1}\sum_{j=1}^{d}\left(K_{\theta_{0}}^{-1}\right)_{ij}\left\{ \left(X_{j}^{(n)}\otimes I_{\H_{a}^{(n)}}\right)\ket{\psi_{\theta_{0}}}\bra{\psi_{\theta_{0}}}-\ket{\psi_{\theta_{0}}}\bra{\psi_{\theta_{0}}}\left(X_{j}^{(n)}\otimes I_{\H_{a}^{(n)}}\right)\right\} 
\end{align}
at $\theta_{0}\in\Theta$ for $1\leq i\leq d$ by using (\ref{eq:sld_pure}),
where $\ket{\psi_{\theta_{0}}}=\left((U_{\theta_{0}})^{\otimes n}\otimes I_{\H_{a}^{(n)}}\right)\ket{\psi}$.
The vector representation of the $i$th SLD is
\begin{align*}
\ket{l_{i}^{(n,\psi)}} & :=L_{\theta_{0},i}^{(S,n,\psi)}\ket{\psi_{\theta_{0}}}\\
 & =\sqrt{-1}\sum_{j=1}^{3}\left(K_{\theta_{0}}^{-1}\right)_{ij}\left(X_{j}^{(n)}\otimes I_{\H_{a}^{(n)}}-\bra{\psi_{\theta_{0}}}\left(X_{j}^{(n)}\otimes I_{\H_{a}^{(n)}}\right)\ket{\psi_{\theta_{0}}}\right)\ket{\psi_{\theta_{0}}}.
\end{align*}
Let $Z_{\theta_{0}}^{(n,\psi)}$ be a $d\times d$ complex matrix
defined by
\begin{align}
Z_{\theta_{0},ij}^{(n,\psi)} & :=\braket{l_{i}^{(n,\psi)}}{l_{j}^{(n,\psi)}}\\
 & =\sum_{1\leq s,t\leq3}\left(K_{\theta_{0}}^{-1}\right)_{jt}\left(K_{\theta_{0}}^{-1}\right)_{is}\tilde{Z}_{st}^{(n,\psi)}\label{eq:z_n}
\end{align}
for $1\leq i,j\leq3$, where
\[
\tilde{Z}_{st}^{(n,\psi)}=\bra{\psi_{\theta_{0}}}X_{t}^{(n)}X_{s}^{(n)}\otimes I_{\H_{a}^{(n)}}\ket{\psi_{\theta_{0}}}-\bra{\psi_{\theta_{0}}}X_{t}^{(n)}\otimes I_{\H_{a}^{(n)}}\ket{\psi_{\theta_{0}}}\bra{\psi_{\theta_{0}}}X_{s}^{(n)}\otimes I_{\H_{a}^{(n)}}\ket{\psi_{\theta_{0}}}.
\]
Because of Theorem \ref{thm:sld_ineq_pure}, a matrix inequality
\begin{equation}
J_{\theta_{0}}^{(C,n,\psi,M)}\leq J_{\theta_{0}}^{(S,n,\psi)}\label{eq:sld_ineq_su2}
\end{equation}
holds, where $J_{\theta_{0}}^{(S,n,\psi)}={\rm Re}Z_{\theta_{0}}^{(n,\psi)}$
is the SLD Fisher information matrix with respect to the input $\ket{\psi}$.
When $Z_{\theta_{0}}^{(n,\psi)}$ is a real matrix, there exists a
POVM $M(\psi)$ that achieves the equality. Let 
\begin{equation}
\A_{\theta_{0}}^{(n)}:=\left\{ \ket{\psi}\in\H^{\otimes n}\otimes\H_{a}^{(n)}\mid Z_{\theta_{0}}^{(n,\psi)}\text{ is a real matrix}\right\} \subset\H^{\otimes n}\otimes\H_{a}^{(n)}\label{eq:Aset}
\end{equation}
be a set of input vectors that can achieve the equality of (\ref{eq:sld_ineq_su2}). 

Let $\ket{e_{i}^{+}},\ket{e_{i}^{-}}$ be normalized eigenvectors
of $X_{i}$ corresponding to the eigenvalues $+1,-1$ for $1\leq i\leq3$.
By using (\ref{eq:pauli_relation}), it can be seen that
\begin{align}
X_{j}\ket{e_{i}^{+}} & =c_{ji}\ket{e_{i}^{-}}\label{eq:eigen1}\\
X_{j}\ket{e_{i}^{-}} & =\bar{c}_{ji}\ket{e_{i}^{+}}\label{eq:eigen2}\\
X_{k}\ket{e_{i}^{+}} & =c_{ki}\ket{e_{i}^{-}}\label{eq:eigen3}\\
X_{k}\ket{e_{i}^{-}} & =\bar{c}_{ki}\ket{e_{i}^{+}}\label{eq:eigen4}
\end{align}
 for $(i,j,k)\in\left\{ (1,2,3),(2,3,1),(3,1,2)\right\} $ with $c_{ji},c_{ki}\in\C$
such that $\left|c_{ji}\right|=\left|c_{ki}\right|=1$ and ${\rm Re}\,\bar{c}_{ji}c_{ki}=0$.
Let
\[
\ket{e_{i}^{n,t}}:=\frac{1}{\sqrt{n!}}\sum_{\sigma\in S_{n}}\bigotimes_{p=1}^{n}\ket{e_{i}^{\sigma(p)}},\qquad e_{i}^{p}=\begin{cases}
e_{i}^{+} & (1\leq p\leq t)\\
e_{i}^{-} & (t+1\leq p\leq n)
\end{cases}
\]
with the permutation group $S_{n}$ of $\{1,\dots,n\}$ and $0\leq k\leq n$.
By using (\ref{eq:eigen1}), (\ref{eq:eigen2}), (\ref{eq:eigen3}),
(\ref{eq:eigen4}) multiple times with some calculations, we can obtain
\begin{equation}
X_{j}^{(n)}\ket{e_{i}^{n,t}}=c_{ji}\sqrt{t(n-t+1)}\ket{e_{i}^{n,(t-1)}}+\bar{c}_{ji}\sqrt{(t+1)(n-t)}\ket{e_{i}^{n,(t+1)}}\label{eq:Xn_1}
\end{equation}
\begin{equation}
X_{k}^{(n)}\ket{e_{i}^{n,t}}=c_{ki}\sqrt{t(n-t+1)}\ket{e_{i}^{n,(t-1)}}+\bar{c}_{ki}\sqrt{(t+1)(n-t)}\ket{e_{i}^{n,(t+1)}},\label{eq:Xn_2}
\end{equation}
and
\begin{equation}
X_{i}^{(n)}\ket{e_{i}^{n,t}}=(2t-n)\ket{e_{i}^{n,t}}.\label{eq:Xn_3}
\end{equation}
By using this, we can obtain
\[
\left(X_{1}^{(n)^{2}}+X_{2}^{(n)^{2}}+X_{3}^{(n)^{2}}\right)\ket{e_{i}^{n,t}}=\left(n^{2}+2n\right)\ket{e_{i}^{n,t}}.
\]
for any $t=0,1,\dots,n.$ $C^{(n)}:=X_{1}^{(n)^{2}}+X_{2}^{(n)^{2}}+X_{3}^{(n)^{2}}$
is known as a Casimir operator. 

\subsection{Matrix inequality for Fisher information matrices\label{subsec:matrix_inequality}}

The real matrix $J_{\theta_{0}}^{(U)}$ can be regarded as a Fisher
information matrix for $SU(2)$ channels, since we can show the following
theorem. 
\begin{thm}
\label{thm:bound_matrix}A real matrix $F\in\F_{\theta_{0}}^{(n)}$
satisfies a matrix inequality
\begin{equation}
F\leq n^{2}J_{\theta_{0}}^{(U)}.\label{eq:ineq_matrix}
\end{equation}

When $d=1$, the equality of (\ref{eq:ineq_matrix}) is achieved by
\begin{equation}
\ket{\psi}=\frac{1}{\sqrt{2}}U_{\theta_{0}}^{*}\left(\ket{e_{1}^{+}}^{\otimes n}+\ket{e_{1}^{+}}^{\otimes n}\right)\label{eq:inputMatD1}
\end{equation}
with unit eigenvectors $\ket{e_{1}^{\pm}}$ of $X_{1}$ corresponding
to eigenvalues $\pm1$ for all $n\in\N$. 

When $n=1$ and $d=1,2,3$, the equality of (\ref{eq:ineq_matrix})
is achieved by a maximally entangled input vector
\begin{equation}
\ket{\psi}=\ket{\psi^{ME}}:=\frac{1}{\sqrt{2}}\left(\ket{e_{+}}\otimes\ket{a_{+}}+\ket{e_{-}}\otimes\ket{a_{-}}\right)\label{eq:inputMatN1}
\end{equation}
with any orthonormal basis $\left\{ \ket{e_{+}},\ket{e_{-}}\right\} $
of the Hilbert space $\H=\C^{2}$, and any orthonormal basis $\left\{ \ket{a_{+}},\ket{a_{-}}\right\} $
of an ancilla Hilbert space $\H_{a}=\C^{2}$. 

When $n=2$ and $d=2$, the equality of (\ref{eq:ineq_matrix}) is
achieved by
\begin{equation}
\ket{\psi}=\frac{1}{\sqrt{2}}U_{\theta_{0}}^{*}\left(\ket{e_{3}^{+}}\otimes\ket{e_{3}^{-}}+\ket{e_{3}^{-}}\otimes\ket{e_{3}^{+}}\right)\label{eq:inputMat22}
\end{equation}
with unit eigenvectors $\ket{e_{3}^{\pm}}$ of $X_{3}$ corresponding
to eigenvalues $\pm1$. 

When $(d-1)n>2$, the matrix inequality (\ref{eq:ineq_matrix}) is
not sharp. 
\end{thm}

\begin{proof}
Due to the matrix inequality (\ref{eq:sld_ineq_pure}), the classical
Fisher information matrix with respect to any channel measurement
$(\psi,M)$ satisfies
\begin{equation}
J_{\theta_{0}}^{(C,n,\psi,M)}\leq J_{\theta_{0}}^{(S,n,\psi)}=\left(K_{\theta_{0}}\right)^{-1}{\rm Re}\tilde{Z}^{(n,\psi)}\left(K_{\theta_{0}}^{*}\right)^{-1}.
\end{equation}
Let $v\in\R^{d}$ be any real unit vector in $\R^{d}$. Since the
maximum and the minimum eigenvalues of $\sum_{i=1}^{d}v_{i}X_{i}^{(n)}$
is $\pm n$, 

\begin{align}
v^{\top}\tilde{Z}^{(n,\psi)}v & =\bra{\psi_{\theta_{0}}}\left(\sum_{i=1}^{d}v_{i}X_{i}^{(n)}\right)^{2}\otimes I_{\H_{a}^{(n)}}\ket{\psi_{\theta_{0}}}-\bra{\psi_{\theta_{0}}}\left(\sum_{i=1}^{d}v_{i}X_{i}^{(n)}\right)\otimes I_{\H_{a}^{(n)}}\ket{\psi_{\theta_{0}}}^{2}\\
 & \leq\bra{\psi_{\theta_{0}}}\left(\sum_{i=1}^{d}v_{i}X_{i}^{(n)}\right)^{2}\otimes I_{\H_{a}^{(n)}}\ket{\psi_{\theta_{0}}}\leq n^{2}.
\end{align}
This implies ${\rm Re}\tilde{Z}^{(n,\psi)}\leq n^{2}I,$ and we have
\begin{equation}
J_{\theta_{0}}^{(S,n,\psi)}\leq n^{2}\left(K_{\theta_{0}}\right)^{-1}\left(K_{\theta_{0}}^{*}\right)^{-1}=n^{2}J_{\theta_{0}}^{(U)}.
\end{equation}
This proves (\ref{eq:ineq_matrix}).

When $d=1$, the input vector (\ref{eq:inputMatD1}) satisfies
\begin{align*}
X_{1}^{(n)}U_{\theta_{0}}^{\otimes n}\ket{\psi} & =\frac{1}{\sqrt{2}}X_{1}^{(n)}\left(\ket{e_{1}^{+}}^{\otimes n}+\ket{e_{1}^{-}}^{\otimes n}\right)\\
 & =\frac{n}{\sqrt{2}}\left(\ket{e_{1}^{+}}^{\otimes n}-\ket{e_{1}^{-}}^{\otimes n}\right).
\end{align*}
By using this, we have
\[
\tilde{Z}^{(n,\psi)}=\bra{\psi}U_{\theta_{0}}^{*\otimes n}X_{1}^{(n)^{2}}U_{\theta_{0}}^{\otimes}\ket{\psi}-\bra{\psi}U_{\theta_{0}}^{*\otimes n}X_{1}^{(n)}U_{\theta_{0}}^{\otimes n}\ket{\psi}^{2}=n^{2}I.
\]
Since $\tilde{Z}^{(n,\psi)}\in\R$, $\psi\in\A_{\theta_{0}}^{(n)}$
and there exists a POVM $M$ such that $J_{\theta_{0}}^{(C,n,\psi,M)}=J_{\theta_{0}}^{(S,n,\psi)}=n^{2}J_{\theta_{0}}^{(U)}$. 

When $n=1$, since the input vector (\ref{eq:inputMatN1}) satisfies
\[
\bra{\psi_{\theta_{0}}^{ME}}X_{t}^{(1)}X_{s}^{(1)}\otimes I_{\H_{a}}\ket{\psi_{\theta_{0}}^{ME}}=\frac{1}{2}\Tr X_{t}^{(1)}X_{s}^{(1)}=\frac{1}{4}\delta_{ts}\quad(1\leq s,t\leq d)
\]
and
\[
\bra{\psi_{\theta_{0}}^{ME}}X_{t}^{(1)}\otimes I_{\H_{a}}\ket{\psi_{\theta_{0}}^{ME}}=0\quad(1\leq t\leq d),
\]
it can be seen that
\begin{equation}
\tilde{Z}^{(1,\psi^{ME})}=I,\label{eq:best_1}
\end{equation}
where $\ket{\psi_{\theta_{0}}^{ME}}:=\left(U_{\theta_{0}}\otimes I\right)\ket{\psi^{ME}}$.
Thus, due to (\ref{eq:z_n}),
\begin{equation}
Z_{\theta_{0}}^{(1,\psi^{ME})}=\left(K_{\theta_{0}}\right)^{-1}\tilde{Z}^{(1,\psi^{ME})}\left(K_{\theta_{0}}^{*}\right)^{-1}=J_{\theta_{0}}^{(U)}.
\end{equation}
Since $Z_{\theta_{0}}^{(1,\psi^{ME})}=I$ is a real matrix, $\psi^{ME}\in\A_{\theta_{0}}^{(1)}$
and the equality of (\ref{eq:ineq_matrix}) is achieved by the channel
measurement $(\psi^{ME},M(\psi^{ME}))$. 

When $n=2$ and $d=2$, the input vector (\ref{eq:inputMat22}) satisfies
\begin{align*}
X_{i}^{(2)}U_{\theta_{0}}^{\otimes2}\ket{\psi} & =\frac{1}{\sqrt{2}}X_{i}^{(2)}\left(\ket{e_{3}^{+}}\otimes\ket{e_{3}^{-}}+\ket{e_{3}^{-}}\otimes\ket{e_{3}^{+}}\right)\\
 & =\frac{2}{\sqrt{2}}\left(\bar{c}_{i3}\ket{e_{3}^{+}}\otimes\ket{e_{3}^{+}}+c_{i3}\ket{e_{3}^{-}}\otimes\ket{e_{3}^{-}}\right)
\end{align*}
for $i=1,2$. By using this, we have $\tilde{Z}^{(2,\psi)}=4I$. Since
$\tilde{Z}^{(2,\psi)}$ is a real matrix, $\psi\in\A_{\theta_{0}}^{(2)}$
and there exists a POVM $M$ such that $J_{\theta_{0}}^{(C,2,\psi,M)}=J_{\theta_{0}}^{(S,2,\psi)}=4I$. 

The non-sharpness of (\ref{eq:ineq_matrix}) for $(d-1)n>2$ can be
given by the following Theorem \ref{thm:bound3}.
\end{proof}
The Fisher information $J_{\theta_{0}}^{(U)}$ for a $SU(2)$ model
is same as a quantum channel Fisher information $J_{QC}$ introduced
in \cite{rev10} if $d=1$. The quantum channel Fisher information
was extended to multi-parameter models by considering the maximization
of SLD Fisher information matrices with respect to input states in
\cite{rev11}. Due to Theorem \ref{thm:bound_matrix}, we can see
there is no maximum of the SLD Fisher information matrices when $(d-1)n>2$,
and $J_{\theta_{0}}^{(U)}$ is not included in \cite{rev11}. As for
a $SU(d)$ channel model with $d\geq3$, we can define a matrix similar
to $J_{\theta_{0}}^{(U)}$, however, it does not have the same properties
as Theorem \ref{thm:bound_matrix} even when $d=1$. For example,
a quantum channel Fisher information $J_{QC}$ of a $SU(3)$ model
$\left\{ \sum_{i=1}^{3}e^{\ii a_{i}\theta}\ket i\bra i\mid\theta\in\R\right\} $
at $\theta=0$ with $a_{1}<a_{2}<a_{3}$ and $a_{1}+a_{2}+a_{3}=0$
is $(a_{3}-a_{1})^{2}$, and it does not depend on $a_{2}$, while
$J_{\theta_{0}}^{(U)}=2(a_{1}^{2}+a_{2}^{2}+a_{3}^{2})$. 

\subsection{Inequality for trace of Fisher information matrices\label{subsec:tr_inquality}}

Next, we show inequalities for inverse $J_{\theta_{0}}^{(U)}$ weighted
traces of a classical Fisher information matrices by using Casimir
operators. These inequalities play an essential role in the estimation
of SU(2) channels.
\begin{thm}
\label{thm:bound3}Let $F\in\F_{\theta_{0}}^{(n)}$ be a $d\times d$
real matrix. When $d=1$,
\begin{equation}
\Tr J_{\theta_{0}}^{(U)^{-1}}F\leq n^{2}.\label{eq:ineqTr1}
\end{equation}
When $d=2$ and $n$ is an even number or $d=3$,
\begin{equation}
\Tr J_{\theta_{0}}^{(U)^{-1}}F\leq n^{2}+2n.\label{eq:ineqTrEven}
\end{equation}
When $d=2$ and $n$ is an odd number,
\begin{equation}
\Tr J_{\theta_{0}}^{(U)^{-1}}F\leq n^{2}+2n-1.\label{eq:ineqTrOdd}
\end{equation}
The inequalities (\ref{eq:ineqTr1}), (\ref{eq:ineqTrEven}), and
(\ref{eq:ineqTrOdd}) are sharp. 
\end{thm}

\begin{proof}
When $d=1$, the inequality (\ref{eq:ineqTr1}) and its sharpness
were proved in Theorem \ref{thm:bound_matrix}. 

When $d=3$, due to the inequality (\ref{eq:sld_ineq_su2}), for any
channel measurement $\left(\psi,M\right)$,
\begin{align*}
\Tr J_{\theta_{0}}^{(U)^{-1}}J_{\theta_{0}}^{(C,n,\psi,M)} & =\Tr K_{\theta_{0}}J_{\theta_{0}}^{(C,n,\psi,M)}K_{\theta_{0}}^{*}\\
 & \leq\Tr K_{\theta_{0}}J_{\theta_{0}}^{(S,n,\psi)}K_{\theta_{0}}^{*}=\Tr\tilde{Z}^{(n,\psi)}\\
 & \leq\bra{\psi}C^{(n)}\otimes I_{\H_{a}^{(n)}}\ket{\psi}\\
 & \leq n^{2}+2n,
\end{align*}
 where $C^{(n)}$ is the Casimir operator. It is known that the maximum
eigenvalue of $C^{(n)}$ is $n^{2}+2n$. When $n\geq3$, we can see
that an input vector 
\[
\ket{\psi}=\frac{1}{\sqrt{2}}U_{\theta_{0}}^{*\otimes n}\left(\ket{e_{1}^{+}}^{\otimes n}+\ket{e_{1}^{-}}^{\otimes n}\right)
\]
with unit eigenvectors $\ket{e_{1}^{\pm}}$ of $X_{1}$ corresponding
to eigenvalues $\pm1$ for all $n\in\N$ satisfies
\[
\tilde{Z}^{(n,\psi)}=\begin{pmatrix}n^{2} & 0 & 0\\
0 & n & 0\\
0 & 0 & n
\end{pmatrix}.
\]
Since $\tilde{Z}^{(n,\psi)}$ is a real matrix and $\Tr\tilde{Z}^{(n,\psi)}=n^{2}+2n$,
this input vector $\ket{\psi}$ can achieve the inequality of (\ref{eq:ineqTrEven}).
When $n=2$, channel measurements to achieve (\ref{eq:ineqTrEven})
is given in Theorem \ref{thm:n2}. 

When $d=2$, for any channel measurement $\left(\psi,M\right)$,
\begin{align*}
\Tr J_{\theta_{0}}^{(U)^{-1}}J_{\theta_{0}}^{(C,n,\psi,M)} & =\Tr K_{\theta_{0}}J_{\theta_{0}}^{(C,n,\psi,M)}K_{\theta_{0}}^{*}\\
 & \leq\Tr K_{\theta_{0}}J_{\theta_{0}}^{(S,n,\psi)}K_{\theta_{0}}^{*}=\Tr\tilde{Z}^{(n,\psi)}\\
 & \leq\bra{\psi}\left(C^{(n)}-X_{3}^{(n)^{2}}\right)\otimes I_{\H_{a}^{(n)}}\ket{\psi}.
\end{align*}
To maximize this, we consider the maximization of $\bra{\psi}C^{(n)}\otimes I_{\H_{a}^{(n)}}\ket{\psi}$
and the minimization of $\bra{\psi}X_{3}^{(n)^{2}}\otimes I_{\H_{a}^{(n)}}\ket{\psi}$.
Since the maximum eigenvalue of $C^{(n)}$ is $n^{2}+2n$, $\max_{\psi}\bra{\psi}C^{(n)}\otimes I_{\H_{a}^{(n)}}\ket{\psi}=n^{2}+2n$.
When $n$ is an even number, the eigenvalues of $X_{3}^{(n)}$ are
$n,n-2,\dots,2,0,-2,\dots,-n+2,-n$, thus $\min_{\psi}\bra{\psi}X_{3}^{(n)^{2}}\otimes I_{\H_{a}^{(n)}}\ket{\psi}=0$.
Therefore $\bra{\psi}\left(C^{(n)}-X_{3}^{(n)^{2}}\right)\otimes I_{\H_{a}^{(n)}}\ket{\psi}\leq n^{2}+2n$,
and the optimal input vector is 
\begin{equation}
\ket{\psi}=U_{\theta_{0}}^{*}\ket{e_{3}^{n,n/2}}.\label{eq:opt_even}
\end{equation}
In fact, by using (\ref{eq:Xn_1}), (\ref{eq:Xn_2}), (\ref{eq:Xn_3}),
it can be seen
\[
X_{j}^{(n)}\ket{e_{3}^{n,n/2}}=c_{j,3}\sqrt{\frac{1}{4}(n^{2}+2n)}\ket{e_{3}^{n,(n/2-1)}}+\bar{c}_{j,3}\sqrt{\frac{1}{4}(n^{2}+2n)}\ket{e_{3}^{n,(n/2+1)}}
\]
for $j=1,2$, and it can be seen that
\[
\tilde{Z}^{(n,\psi)}=\frac{1}{2}(n^{2}+2n)I
\]
due to $\re c_{1,3}\bar{c}_{2,3}=0$. Since $\tilde{Z}^{(n,\psi)}$
is a real matrix, there exists a POVM $M$ such that $J_{\theta_{0}}^{(C,n,\psi,M)}=J_{\theta_{0}}^{(S,n,\psi)}$.
When $n$ is an odd number, the eigenvalues of $X_{3}^{(n)}$ are
$n,n-2,\dots,1,-1,\dots,-n+2,-n$, thus $\min_{\psi}\bra{\psi}X_{3}^{(n)^{2}}\otimes I_{\H_{a}^{(n)}}\ket{\psi}=1$.
Therefore $\bra{\psi}\left(C^{(n)}-X_{3}^{(n)^{2}}\right)\otimes I_{\H_{a}^{(n)}}\ket{\psi}\leq n^{2}+2n-1$,
and the optimal input vector is 
\begin{equation}
\ket{\psi}=\frac{1}{\sqrt{2}}\left(U_{\theta_{0}}^{*\otimes n}\otimes I\right)\left\{ \ket{e_{3}^{n,(n+1)/2}}\ket{a_{+}}+\ket{e_{3}^{n,(n-1)/2}}\ket{a_{-}}\right\} ,\label{eq:opt_odd}
\end{equation}
with an orthonormal basis $\ket{a_{\pm}}$ of an ancilla Hilbert space
$\H_{a}\simeq\C^{2}$. In fact, by using (\ref{eq:Xn_1}), (\ref{eq:Xn_2}),
(\ref{eq:Xn_3}), it can be seen
\begin{align*}
X_{j}^{(n)}\left(U_{\theta_{0}}^{\otimes n}\otimes I\right)\ket{\psi} & =\left\{ c_{j,3}\frac{n+1}{2\sqrt{2}}\ket{e_{3}^{n,(n-1)/2}}+\bar{c}_{j,3}\sqrt{\frac{1}{8}(n+3)(n-1)}\ket{e_{3}^{n,(n+3)/2}}\right\} \ket{a_{+}}\\
 & \qquad+\left\{ \bar{c}_{j,3}\frac{n+1}{2\sqrt{2}}\ket{e_{3}^{n,(n+1)/2}}+c_{j,3}\sqrt{\frac{1}{8}(n+3)(n-1)}\ket{e_{3}^{n,(n-3)/2}}\right\} \ket{a_{-}}
\end{align*}
for $j=1,2$, and it can be seen that
\[
\tilde{Z}^{(n,\psi)}=\frac{1}{2}(n^{2}+2n-1)I
\]
due to $\re c_{1,3}\bar{c}_{2,3}=0$. Since $\tilde{Z}^{(n,\psi)}$
is a real matrix, there exists a POVM $M$ such that $J_{\theta_{0}}^{(C,n,\psi,M)}=J_{\theta_{0}}^{(S,n,\psi)}$. 
\end{proof}
Imai and Fujiwara \cite{imai_sud} showed a similar inequality to
this theorem for the estimation of $SU(\dim\H)$ channels with respect
to group covariant weights. On the other hand, this inequality has
a very similar form to the inequality introduced by Gill and Massar
for the estimation of qubit states \cite{gill_massar,yama_tomo}.
In the following subsection, we show the lower bound of weighted traces
of inverse Fisher information matrices for the estimation of $SU(2)$
channels with arbitrary weights, inspired by Gill and Massar's method.

\subsection{Three parameter estimation\label{subsec:3para}}

When $d=3$, we obtain Gill and Massar type bound as follows.
\begin{thm}
\label{thm:gill_massar_single}When $d=3$, for any $3\times3$ positive
real matrix $W$, a real matrix $F\in\F_{\theta_{0}}^{(n)}$ satisfies
the inequality
\begin{equation}
\Tr WF^{-1}\geq\frac{1}{n^{2}+2n}\left(\Tr\sqrt{\tilde{W}}\right)^{2}\label{eq:gill_massar}
\end{equation}
 where $\tilde{W}:=\sqrt{J_{\theta_{0}}^{(U)^{-1}}}W\sqrt{J_{\theta_{0}}^{(U)^{-1}}}$.
The equality is achieved if and only if
\begin{equation}
F=\frac{n^{2}+2n}{\Tr\sqrt{\tilde{W}}}\sqrt{J_{\theta_{0}}^{(U)}}\sqrt{\tilde{W}}\sqrt{J_{\theta_{0}}^{(U)}}.\label{eq:gill_massar_cond}
\end{equation}
\end{thm}

\begin{proof}
Let $\tilde{F}:=\sqrt{J_{\theta_{0}}^{(U)^{-1}}}F\sqrt{J_{\theta_{0}}^{(U)^{-1}}}$.
Since $\frac{1}{n^{2}+2n}\Tr\tilde{F}\leq1$ due to Theorem \ref{thm:bound3},
\begin{align}
\Tr WF^{-1} & =\Tr\tilde{W}\tilde{F}^{-1}\\
 & \geq\left(\frac{1}{n^{2}+2n}\Tr\tilde{F}\right)\left(\Tr\sqrt{\tilde{F}^{-1}}\tilde{W}\sqrt{\tilde{F}^{-1}}\right)\label{eq:gill_masser_su2_proof1}\\
 & \geq\frac{1}{n^{2}+2n}\left(\Tr\left\{ \sqrt{\tilde{F}}\right\} \left\{ \sqrt{\tilde{W}}\sqrt{\tilde{F}^{-1}}\right\} \right)^{2}\label{eq:gill_masser_su2_proof2}\\
 & =\frac{1}{n^{2}+2n}\left(\Tr\sqrt{\tilde{W}}\right)^{2},
\end{align}
where the Cauchy-Schwarz inequality is used in the second inequality,
and the equality is achieved if and only if $\sqrt{\tilde{F}}=k\sqrt{\tilde{W}}\sqrt{\tilde{F}^{-1}}$
with a real value $k$. The equality of (\ref{eq:gill_masser_su2_proof1})
is achieved when $k=\frac{n^{2}+2n}{\Tr\sqrt{\tilde{W}}}$. 
\end{proof}
In the following theorems, we construct concrete randomized channel
measurement that can achieve the inequality (\ref{eq:gill_massar}).
Let us prepare observables $\left\{ X_{i}\right\} _{i=1}^{3}$ defined
at (\ref{eq:X_def}) with a $3\times3$ real orthogonal matrix $O$
such that
\[
O\tilde{W}O^{*}={\rm diag}(w_{1},w_{2},w_{3}),
\]
and $0\leq w_{1}\leq w_{2}\leq w_{3}$.
\begin{thm}
\label{thm:optM3n}When
\begin{equation}
n\geq\max\left\{ 3,\frac{\sqrt{w_{2}}+\sqrt{w_{3}}}{\sqrt{w_{1}}}-1\right\} ,\label{eq:ineq_n_cond}
\end{equation}
the equality of (\ref{eq:gill_massar}) is achievable by a randomized
channel measurement
\[
\bigoplus_{i=1}^{3}s_{i}\left(\psi_{i}^{(n)},M(\psi_{i}^{(n)})\right),
\]
where
\[
\ket{\psi_{i}^{(n)}}:=\frac{1}{\sqrt{2}}U_{\theta_{0}}^{*\otimes n}\left(\ket{e_{i}^{+}}^{\otimes n}+\ket{e_{i}^{-}}^{\otimes n}\right)\in\H^{\otimes n}\qquad(1\leq i\leq3),
\]
and 
\[
s_{i}=\frac{(n+2)\frac{\sqrt{w_{i}}}{\sqrt{w_{1}}+\sqrt{w_{2}}+\sqrt{w_{3}}}-1}{n-1}.
\]
\end{thm}

\begin{proof}
When $n\geq3$, by using (\ref{eq:eigen1}), (\ref{eq:eigen2}), (\ref{eq:eigen3}),
and (\ref{eq:eigen4}), it can be seen that
\begin{align}
X_{j}^{(n)}U_{\theta_{0}}^{\otimes n}\ket{\psi_{i}^{(n)}} & =\frac{1}{\sqrt{2}}\sum_{k=1}^{n}\left(c_{ji}\ket{e_{i}^{+}}^{\otimes(k-1)}\otimes\ket{e_{i}^{-}}\otimes\ket{e_{i}^{+}}^{\otimes(n-1)}+\bar{c}_{ji}\ket{e_{i}^{-}}^{\otimes(k-n)}\otimes\ket{e_{i}^{+}}\otimes\ket{e_{i}^{-}}^{\otimes(n-k)}\right),\label{eq:xnph_j}\\
X_{k}^{(n)}U_{\theta_{0}}^{\otimes n}\ket{\psi_{i}^{(n)}} & =\frac{1}{\sqrt{2}}\sum_{k=1}^{n}\left(c_{ki}\ket{e_{i}^{+}}^{\otimes(k-1)}\otimes\ket{e_{i}^{-}}\otimes\ket{e_{i}^{+}}^{\otimes(n-1)}+\bar{c}_{ki}\ket{e_{i}^{-}}^{\otimes(k-n)}\otimes\ket{e_{i}^{+}}\otimes\ket{e_{i}^{-}}^{\otimes(n-k)}\right),\label{eq:xnph_k}\\
X_{i}^{(n)}U_{\theta_{0}}^{\otimes n}\ket{\psi_{i}^{(n)}} & =\frac{1}{\sqrt{2}}\left(n\ket{e_{i}^{+}}^{\otimes n}-n\ket{e_{i}^{-}}^{\otimes n}\right).\label{eq:xnph_i}
\end{align}
for $(i,j,k)\in\left\{ (1,2,3),(2,3,1),(3,1,2)\right\} $, and we
obtain
\begin{equation}
\left[\bra{\psi_{i}^{(n)}}U_{\theta_{0}}^{*\otimes n}X_{\alpha}^{(n)}X_{\beta}^{(n)}U_{\theta_{0}}^{\otimes n}\ket{\psi_{i}^{(n)}}\right]_{1\leq\alpha,\beta\leq3}=nI+(n^{2}-n)\ket{e_{i}}\bra{e_{i}}.\label{eq:best_n}
\end{equation}
Thus, due to (\ref{eq:z_n}),
\[
Z_{\theta_{0}}^{(n,\psi_{i}^{(n)})}=\left(K_{\theta_{0}}\right)^{-1}\left(nI+(n^{2}-n)\ket{e_{i}}\bra{e_{i}}\right)\left(K_{\theta_{0}}^{*}\right)^{-1}.
\]
Since $Z_{\theta_{0}}^{(n,\psi_{i}^{(n)})}$ is a real matrix, $\psi_{i}^{(n)}\in\A_{\theta_{0}}^{(n)}$
and the classical Fisher information matrix with respect to a channel
measurement $(\psi_{i}^{(n)},M(\psi_{i}^{(n)}))$ is
\[
J_{\theta_{0}}^{(C,n,\psi_{i}^{(n)},M(\psi_{i}^{(n)}))}=\left(K_{\theta_{0}}\right)^{-1}\left(nI+(n^{2}-n)\ket{e_{i}}\bra{e_{i}}\right)\left(K_{\theta_{0}}^{*}\right)^{-1}.
\]

The classical Fisher information matrix with respect to the randomized
channel measurement $\bigoplus_{i=1}^{3}s_{i}\left(\psi_{i}^{(n)},M(\psi_{i}^{(n)})\right)$is
\begin{align*}
J_{\theta_{0}}^{(C,n,\bigoplus_{i=1}^{3}s_{i}\left(\psi_{i}^{(n)},M(\psi_{i}^{(n)})\right))} & =\left(K_{\theta_{0}}\right)^{-1}\left\{ nI+(n^{2}-n)\sum_{i=1}^{3}s_{i}\ket{e_{i}}\bra{e_{i}}\right\} \left(K_{\theta_{0}}^{*}\right)^{-1}\\
 & =\sqrt{J_{\theta_{0}}^{(U)}}\left\{ nI+(n^{2}-n)\sum_{i=1}^{3}s_{i}\ket{w_{i}}\bra{w_{i}}\right\} \sqrt{J_{\theta_{0}}^{(U)}}
\end{align*}
and this is identical to the optimal $F$ in Theorem \ref{thm:gill_massar_single}.
Note that $n\geq\frac{\sqrt{w_{2}}+\sqrt{w_{3}}}{\sqrt{w_{1}}}-1$
implies $s_{i}\geq0$ for $1\leq i\leq3$. 
\end{proof}
When $n\geq3$ and $n<\frac{\sqrt{w_{2}}+\sqrt{w_{3}}}{\sqrt{w_{1}}}-1$,
it is hard to obtain analytical optimal channel measurements. However,
by taking $n$ large enough, condition (\ref{eq:ineq_n_cond}) can
be always satisfied for any $W>0$. If input states (\ref{eq:opt_even})
and (\ref{eq:opt_odd}) are incorporated into random measurements,
the range (\ref{eq:ineq_n_cond}) can be expanded slightly.

When $n=2$, $X_{j}^{(n)}U_{\theta_{0}}^{\otimes n}\ket{\psi_{i}^{(n)}}$
and $X_{k}^{(n)}U_{\theta_{0}}^{\otimes n}\ket{\psi_{i}^{(n)}}$ at
(\ref{eq:xnph_j}) and (\ref{eq:xnph_k}) are not orthogonal, and
(\ref{eq:best_n}) is not satisfied. Therefore, this theorem is not
valid. 

In multi-parameter quantum metrology, some researchs reported that
simultaneous estimation of several parameters is better than doing
individual experiments for each parameter \cite{rev6,rev7,rev8,rev9}.
Theorem \ref{thm:optM3n} shows the lower bound is achievable by using
a randomized strategy, and it seems to negate the advantage of simultaneous
estimation over individual estimation. The reason the randomised strategy
can be optimal is that the optimal strategies for the individual parameters
in this theorem also acquire information on other parameters. In fact,
the eigenvalues of (\ref{eq:best_n}) are $\{n^{2},n,n\}$, not $\{n^{2},0,0\}$.

\subsection{Two-parameter Estimation \label{subsec:2para}}

When $d=2$, we also obtain Gill and Massar type bound as follows.
\begin{thm}
\label{thm:gill_massar2}When $d=2$, for any $2\times2$ positive
real matrix $W$, a real matrix $F\in\F_{\theta_{0}}^{(n)}$ satisfies
the inequality
\begin{equation}
\Tr WF^{-1}\geq\begin{cases}
\frac{1}{n^{2}+2n}\left(\Tr\sqrt{\tilde{W}}\right)^{2} & \text{if }n\text{ is even,}\\
\frac{1}{n^{2}+2n-1}\left(\Tr\sqrt{\tilde{W}}\right)^{2} & \text{if }n\text{ is odd.}
\end{cases}\label{eq:gill_massar2}
\end{equation}
 where $\tilde{W}:=\sqrt{J_{\theta_{0}}^{(U)^{-1}}}W\sqrt{J_{\theta_{0}}^{(U)^{-1}}}$. 

When $W=J_{\theta_{0}}^{(U)}$, this equality is achieved by a input
state
\[
\ket{\psi}=\begin{cases}
U_{\theta_{0}}^{*}\ket{e_{3}^{n,n/2}} & \text{if }n\text{ is even,}\\
\frac{1}{\sqrt{2}}\left(U_{\theta_{0}}^{*\otimes n}\otimes I\right)\left\{ \ket{e_{3}^{n,(n+1)/2}}\ket{a_{+}}+\ket{e_{3}^{n,(n-1)/2}}\ket{a_{-}}\right\}  & \text{if }n\text{ is odd,}
\end{cases}
\]
with an orthonormal basis $\ket{a_{\pm}}$ of an ancilla Hilbert space
$\H_{a}\simeq\C^{2}$. 
\end{thm}

\begin{proof}
The proof of (\ref{eq:gill_massar2}) is similar to the proof of Theorem
\ref{thm:gill_massar_single}. The optimal input states for $W=J_{\theta_{0}}^{(U)}$
are given at (\ref{eq:opt_even}) and (\ref{eq:opt_odd}). 
\end{proof}
For an arbitrary weight matrix $W$, it is hard to obtain analytical
optimal strategies when $d=2$. However, we can obtain an asymptotically
optimal sequence of strategies to achieve the Gill and Massar type
bound as follows.

\begin{thm}
For any $2\times2$ real positive matrix $W$, a sequence of randomized
strategies $\bigoplus_{i=1}^{2}s_{i}\left(\psi_{i}^{(n)},M(\psi_{i}^{(n)})\right)$
with
\[
\ket{\psi_{i}^{(n)}}=\frac{1}{\sqrt{2}}U_{\theta_{0}}^{*\otimes n}\left(\ket{e_{i}^{+}}^{\otimes n}+\ket{e_{i}^{-}}^{\otimes n}\right)
\]
and
\[
s_{i}=\frac{\sqrt{w_{i}}}{\sqrt{w_{1}}+\sqrt{w_{2}}}
\]
for $i=1,2$ satisfies
\[
\lim_{n\to\infty}(n^{2}+2n)\Tr WJ_{\theta_{0}}^{(C,n,\bigoplus_{i=1}^{2}s_{i}\left(\psi_{i}^{(n)},M(\psi_{i}^{(n)})\right))^{-1}}=\left(\Tr\sqrt{\tilde{W}}\right)^{2},
\]
where $\tilde{W}:=\sqrt{J_{\theta_{0}}^{(U)^{-1}}}W\sqrt{J_{\theta_{0}}^{(U)^{-1}}}$.
\end{thm}

The proof is similar to the proof of Theorem \ref{thm:optM3n}. 

\subsection{Three parameter estimation for $n=2$\label{subsec:estN2}}

When $d=3$ and $n=2$, we see that Theorem \ref{thm:optM3n} is not
valid. Instead, more informative inequality can be derivatived by
considering both inequality (\ref{eq:ineq_matrix}) and inequality
(\ref{eq:ineqTrEven}) as follows. 
\begin{thm}
\label{thm:gill_massar_both}When $d=3$ and $n=2$, any real matrix
$F\in\F_{\theta_{0}}^{(2)}$ satisfies
\begin{equation}
\Tr WF^{-1}\geq\begin{cases}
\frac{1}{8}\left(\Tr\sqrt{\tilde{W}}\right)^{2} & \text{if }\frac{\sqrt{w_{3}}}{\sqrt{w_{1}}+\sqrt{w_{2}}+\sqrt{w_{3}}}<\frac{1}{2},\\
\frac{1}{4}\left(\sqrt{w_{1}}+\sqrt{w_{2}}\right)^{2}+\frac{w_{3}}{4} & \text{if }\frac{\sqrt{w_{3}}}{\sqrt{w_{1}}+\sqrt{w_{2}}+\sqrt{w_{3}}}\geq\frac{1}{2}.
\end{cases}\label{eq:gill_massar_32}
\end{equation}
When $\frac{\sqrt{w_{3}}}{\sqrt{w_{1}}+\sqrt{w_{2}}+\sqrt{w_{3}}}<\frac{1}{2}$,
the equality is achieved if and only if
\begin{equation}
F=\frac{8}{\Tr\sqrt{\tilde{W}}}\sqrt{J_{\theta_{0}}^{(U)}}\sqrt{\tilde{W}}\sqrt{J_{\theta_{0}}^{(U)}}.\label{eq:gill_massar_both_cond1}
\end{equation}
When $\frac{\sqrt{w_{3}}}{\sqrt{w_{1}}+\sqrt{w_{2}}+\sqrt{w_{3}}}\geq\frac{1}{2}$,
the equality is achieved if and only if
\begin{equation}
F=\frac{4}{\sqrt{w_{1}}+\sqrt{w_{2}}}\sqrt{J_{\theta_{0}}^{(U)}}\left\{ \left(\sqrt{w_{1}}\ket{w_{1}}\bra{w_{1}}+\sqrt{w_{2}}\ket{w_{2}}\bra{w_{2}}\right)+4\ket{w_{3}}\bra{w_{3}}\right\} \sqrt{J_{\theta_{0}}^{(U)}}.\label{eq:gill_massar_both_cond2}
\end{equation}
Where $w_{1},w_{2},w_{3}$ are eigenvalues of $\tilde{W}$ such that
$0\leq w_{1}\leq w_{2}\leq w_{3}$, and $\ket{w_{1}},\ket{w_{2}},\ket{w_{3}}\in\R^{3}$
are corresponding orthonormal eigenvectors. 
\end{thm}

\begin{proof}
When $\frac{\sqrt{w_{3}}}{\sqrt{w_{1}}+\sqrt{w_{2}}+\sqrt{w_{3}}}<\frac{1}{2}$,
the inequality (\ref{eq:gill_massar_32}) and the condition (\ref{eq:gill_massar_both_cond1})
were obtained in Theorem \ref{thm:gill_massar_single}. Note that
$F=\frac{8}{\Tr\sqrt{\tilde{W}}}\sqrt{J_{\theta_{0}}^{(U)}}\sqrt{\tilde{W}}\sqrt{J_{\theta_{0}}^{(U)}}\leq4J_{\theta_{0}}^{(U)}$
since $2\sqrt{\tilde{W}}\leq2\sqrt{w_{3}}I\leq\Tr\sqrt{\tilde{W}}I$. 

When $\frac{\sqrt{w_{3}}}{\sqrt{w_{1}}+\sqrt{w_{2}}+\sqrt{w_{3}}}\geq\frac{1}{2}$,
by considering both inequality (\ref{eq:ineq_matrix}) and inequality
(\ref{eq:ineqTrEven}), we have
\begin{align}
\Tr WF^{-1} & \geq\min\left\{ \Tr WF_{0}^{-1}\mid F_{0}\text{ is a \ensuremath{3\times3} real matrix},0<F_{0}\leq4J_{\theta_{0}}^{(U)},\Tr J_{\theta_{0}}^{(U)^{-1}}F_{0}\leq8\right\} \label{eq:optimal_double1}\\
 & =\min\left\{ \Tr\tilde{W}\tilde{F_{0}}^{-1}\mid\tilde{F_{0}}\text{ is a \ensuremath{3\times3} real matrix},0<\tilde{F_{0}}\leq4I,\Tr\tilde{F_{0}}\leq8\right\} \label{eq:optimal_double2}\\
 & =\min\left\{ \Tr{\rm diag}(w_{1},w_{2},w_{3})\tilde{\tilde{F}}_{0}^{-1}\mid\tilde{\tilde{F}}_{0}\text{ is a \ensuremath{3\times3} real matrix},0<\tilde{\tilde{F}}_{0}\leq4I,\Tr\tilde{\tilde{F}}_{0}\leq8\right\} .\label{eq:optimal_double3}
\end{align}
Since the $3\times3$ real matrix $\tilde{\tilde{F}}_{0}$ in (\ref{eq:optimal_double3})
can be partitioned as
\[
\tilde{\tilde{F}}_{0}=\begin{pmatrix}A & \ket b\\
\bra b & c
\end{pmatrix}
\]
with a $2\times2$ positive real matrix $A$ and a real vector $\ket b\in\R^{2}$
and a positive real number $c\in\R$, 
\begin{align*}
\Tr{\rm diag}(w_{1},w_{2},w_{3})\tilde{\tilde{F}}^{-1} & =\Tr{\rm diag}(w_{1},w_{2})\left(A-\frac{1}{c}\ket b\bra b\right)^{-1}+w_{3}\left(c-\bra bA^{-1}\ket b\right)^{-1}\\
 & \geq\Tr{\rm diag}(w_{1},w_{2})A^{-1}+w_{3}c^{-1}.
\end{align*}
Therefore,
\begin{align}
\text{(\ref{eq:optimal_double3})} & \geq\min_{A,c}\left\{ \Tr{\rm diag}(w_{1},w_{2})A^{-1}+w_{3}c^{-1}\mid0<A\leq4I,\,\Tr A\leq8-c,\,0<c\leq4\right\} \label{eq:optimal_double4}\\
 & \geq\min_{A,c}\left\{ \Tr{\rm diag}(w_{1},w_{2})A^{-1}+w_{3}c^{-1}\mid\Tr A\leq8-c,\,0<c\leq4\right\} \label{eq:optimal_double5}\\
 & =\min_{c}\left\{ \frac{1}{8-c}\left(\sqrt{w_{1}}+\sqrt{w_{2}}\right)^{2}+w_{3}c^{-1}\mid\,0<c\leq4\right\} .\label{eq:optimal_double6}
\end{align}
In Equation (\ref{eq:optimal_double6}), the similar argument as Theorem
\ref{thm:gill_massar_single} is used, and the optimal $A$ is
\begin{equation}
A(c)=(8-c)\frac{\sqrt{{\rm diag}(w_{1},w_{2})}}{\Tr\sqrt{{\rm diag}(w_{1},w_{2})}}.\label{eq:optimal_A}
\end{equation}
By considering the minimization of a convex function
\begin{equation}
f:c\mapsto\frac{1}{8-c}\left(\sqrt{w_{1}}+\sqrt{w_{2}}\right)^{2}+w_{3}c^{-1}\qquad(0<c\leq4),\label{eq:convex_func}
\end{equation}
we can obtain the optimal $c$ in (\ref{eq:optimal_double6}) as
\[
c=\min\left\{ 4,8\frac{\sqrt{w_{3}}}{\sqrt{w_{1}}+\sqrt{w_{2}}+\sqrt{w_{3}}}\right\} =4.
\]
The optimal $A$ is
\begin{align*}
A(4) & =4\frac{\sqrt{{\rm diag}(w_{1},w_{2})}}{\Tr\sqrt{{\rm diag}(w_{1},w_{2})}}.
\end{align*}
Thus, the optimal $\tilde{F}_{0}$ in (\ref{eq:optimal_double2})
is
\[
\tilde{F}_{0}=\frac{4}{\sqrt{w_{1}}+\sqrt{w_{2}}}\left(\sqrt{w_{1}}\ket{w_{1}}\bra{w_{1}}+\sqrt{w_{2}}\ket{w_{2}}\bra{w_{2}}\right)+4\ket{w_{3}}\bra{w_{3}},
\]
and the optimal $F_{0}$ in (\ref{eq:optimal_double1}) is $F_{0}=\sqrt{J_{\theta_{0}}^{(U)}}\tilde{F}_{0}\sqrt{J_{\theta_{0}}^{(U)}}$.
The minimum of (\ref{eq:convex_func}) is
\[
f(4)=\frac{1}{4}\left(\sqrt{w_{1}}+\sqrt{w_{2}}\right)^{2}+\frac{w_{3}}{4}.
\]
\end{proof}
In the following theorems, we construct a concrete randomized channel
measurement that achieve the inequality (\ref{eq:gill_massar_32}).
Let us prepare observables $\left\{ X_{i}\right\} _{i=1}^{3}$ defined
at (\ref{eq:X_def}) with a $3\times3$ real orthogonal matrix $O$
such that
\[
O\tilde{W}O^{*}={\rm diag}(w_{1},w_{2},w_{3}).
\]

\begin{thm}
\label{thm:n2}When $d=3$ and $n=2$, the equality of (\ref{eq:gill_massar_32})
is achievable by a randomized channel measurement
\[
\bigoplus_{i=1}^{3}s_{i}\left(\psi_{i}^{(2)},M\left(\psi_{i}^{(2)}\right)\right),
\]
where
\[
\ket{\psi_{i}^{(2)}}:=U_{\theta_{0}}^{*\otimes2}\ket{e_{i}^{2,1}}=\frac{1}{\sqrt{2}}U_{\theta_{0}}^{*\otimes2}\left(\ket{e_{i}^{+}}\otimes\ket{e_{i}^{-}}+\ket{e_{i}^{-}}\otimes\ket{e_{i}^{-}}\right)\in\H^{\otimes2}\qquad(1\leq i\leq3),
\]
and $s_{1},s_{2},s_{3}$ are given as follows. 

When $\frac{\sqrt{w_{3}}}{\sqrt{w_{1}}+\sqrt{w_{2}}+\sqrt{w_{3}}}<\frac{1}{2}$,
\begin{align*}
s_{i} & =\frac{\sqrt{w_{i}}}{\sqrt{w_{1}}+\sqrt{w_{2}}+\sqrt{w_{3}}}\qquad(1\leq i\leq3).
\end{align*}
When $\frac{\sqrt{w_{3}}}{\sqrt{w_{1}}+\sqrt{w_{2}}+\sqrt{w_{3}}}\geq\frac{1}{2}$,
\begin{align*}
s_{1} & =\frac{\sqrt{w_{2}}}{\sqrt{w_{1}}+\sqrt{w_{2}}},\\
s_{2} & =\frac{\sqrt{w_{1}}}{\sqrt{w_{1}}+\sqrt{w_{2}}},\\
s_{3} & =0.
\end{align*}
\end{thm}

\begin{proof}
By using (\ref{eq:eigen1}), (\ref{eq:eigen2}), (\ref{eq:eigen3}),
and (\ref{eq:eigen4}), it can be seen that
\begin{align*}
X_{j}^{(2)}U_{\theta_{0}}^{\otimes2}\ket{\psi_{i}^{(2)}} & =\frac{2}{\sqrt{2}}\left(\bar{c}_{ji}\ket{e_{i}^{+}}\otimes\ket{e_{i}^{+}}+c_{ji}\ket{e_{i}^{-}}\otimes\ket{e_{i}^{-}}\right),\\
X_{k}^{(2)}U_{\theta_{0}}^{\otimes2}\ket{\psi_{i}^{(2)}} & =\frac{2}{\sqrt{2}}\left(\bar{c}_{ki}\ket{e_{i}^{+}}\otimes\ket{e_{i}^{+}}+c_{ki}\ket{e_{i}^{-}}\otimes\ket{e_{i}^{-}}\right),\\
X_{i}^{(2)}U_{\theta_{0}}^{\otimes2}\ket{\psi_{i}^{(2)}} & =0,
\end{align*}
for $(i,j,k)\in\left\{ (1,2,3),(2,3,1),(3,1,2)\right\} $, and we
obtain
\begin{equation}
\tilde{Z}^{(2,\psi_{i}^{(2)})}=\left[\bra{\psi_{i}^{(2)}}U_{\theta_{0}}^{*\otimes2}X_{\alpha}^{(2)}X_{\beta}^{(2)}U_{\theta_{0}}^{\otimes2}\ket{\psi_{i}^{(2)}}\right]_{1\leq\alpha,\beta\leq3}=4\left(I-\ket{e_{i}}\bra{e_{i}}\right),\label{eq:best_2-1}
\end{equation}
where $\ket{e_{1}}=\left(1,0,0\right)^{\top}$, $\ket{e_{2}}=\left(0,1,0\right)^{\top}$,
$\ket{e_{3}}=\left(0,0,1\right)^{\top}$. Thus, due to (\ref{eq:z_n}),
\[
Z_{\theta_{0}}^{(2,\psi_{i}^{(2)})}=\left(K_{\theta_{0}}\right)^{-1}4\left(I-\ket{e_{i}}\bra{e_{i}}\right)\left(K_{\theta_{0}}^{*}\right)^{-1}.
\]
Since $Z_{\theta_{0}}^{(2,\psi_{i}^{(2)})}$ is a real matrix, $\psi_{i}^{(2)}\in\A_{\theta_{0}}^{(n)}$
and the classical Fisher information matrix with respect to the channel
measurement $(\psi_{i}^{(2)},M\left(\psi_{i}^{(2)}\right))$ is
\[
J_{\theta_{0}}^{(C,2,\psi_{i}^{(2)},M\left(\psi_{i}^{(2)}\right))}=Z_{\theta_{0}}^{(2,\psi_{i}^{(2)})}=\left(K_{\theta_{0}}\right)^{-1}4\left(I-\ket{e_{i}}\bra{e_{i}}\right)\left(K_{\theta_{0}}^{*}\right)^{-1}
\]
for $1\leq i\leq3$. 

The classical Fisher information matrix with respect to the randomized
channel measurement
\[
\bigoplus_{i=1}^{3}s_{i}\left(\psi_{i}^{(2)},M(\psi_{i}^{(2)})\right)
\]
is
\begin{align*}
J_{\theta_{0}}^{(C,2,\bigoplus_{i=1}^{3}s_{i}\left(\psi_{i}^{(2)},M(\psi_{i}^{(2)})\right))} & =4\left(K_{\theta_{0}}\right)^{-1}\left\{ I-\sum_{i=1}^{3}s_{i}\ket{e_{i}}\bra{e_{i}}\right\} \left(K_{\theta_{0}}^{*}\right)^{-1}\\
 & =4\sqrt{J_{\theta_{0}}^{(U)}}\left\{ I-\sum_{i=1}^{3}s_{i}\ket{w_{i}}\bra{w_{i}}\right\} \sqrt{J_{\theta_{0}}^{(U)}},
\end{align*}
and this is identical to the optimal $F$ in Theorem \ref{thm:gill_massar_both}. 
\end{proof}

\section{Plots of classical Fisher information matrices\label{sec:plots}}

To visualize the set $\F_{\theta_{0}}^{(n)}$ of classical Fisher
information matrices of three-dimensional parametric unitary channel
model $\left\{ \Gamma_{\theta}^{\otimes n}:\B(\C^{2})\to\B(\C^{2})\mid\,\theta\in\Theta\subset\R^{3}\right\} $,
in Figure \ref{fig:fishers}, we show boundaries of 
\[
\tilde{\F}_{\theta_{0}}^{(n)}:=\left\{ (F_{11},F_{22},F_{33})\mid\,F\in J^{(U)^{-1/2}}\F_{\theta_{0}}^{(n)}J^{(U)^{-1/2}},\,\Tr F=n^{2}+2n\right\} 
\]
 for $n=2$ (top-left), $n=3$ (top-right), $n=4$ (bottom-left),
and $n=5$ (bottom-right). The outside triangle's vertices correspond
to $(F_{11},F_{22},F_{33})=(n^{2}+2n,0,0),(0,n^{2}+2n,0),(0,0,n^{2}+2n)$.
The inner triangles painted in gray indicate classical Fisher information
matrices that can be realized by using the randomized channel measurements
given in Theorem \ref{thm:optM3n} for $n=3,4,5$. For $n=2$, the
inner triangle is same as $\tilde{\F}_{\theta_{0}}^{(n)}$, and it
is realized by using randomized channel measurements given in Theorem
\ref{thm:n2}. The curves surrounding the inner triangles are the
boundaries of $\tilde{\F}_{\theta_{0}}^{(n)}$ obtained by numerical
calculations for $n=3,4,5$. These curves imply that there exist channel
measurements to achieve the equality of (\ref{eq:gill_massar}) without
satisfying (\ref{eq:ineq_n_cond}). See Appendix \ref{sec:numerical_calculations}
for details about the numerical calculations to obtain the curves.
The three black dots correspond to optimal channel measurements when
$d=2$ and weight matrix $W=J_{\theta_{0}}^{(U)}$ given in Theorem
\ref{thm:gill_massar2}. 

\begin{figure}
\begin{centering}
\includegraphics[scale=0.45]{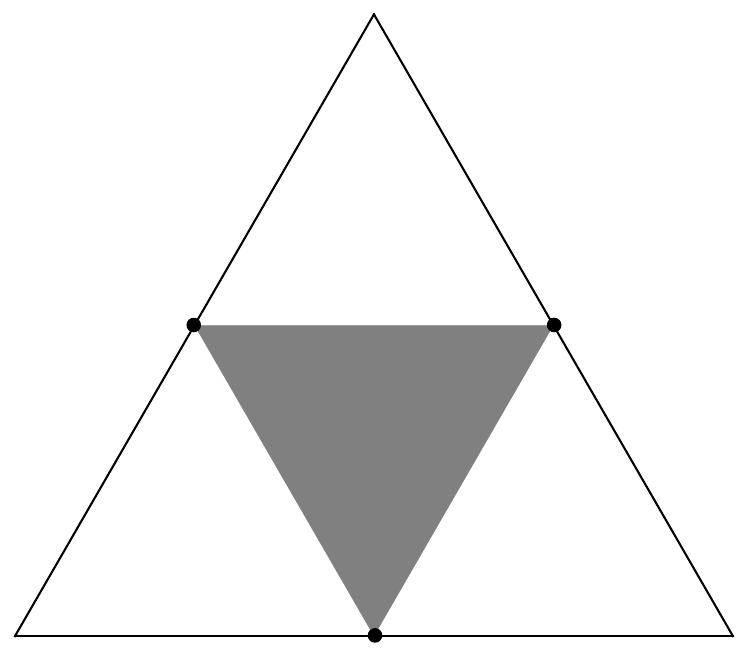}\includegraphics[scale=0.5]{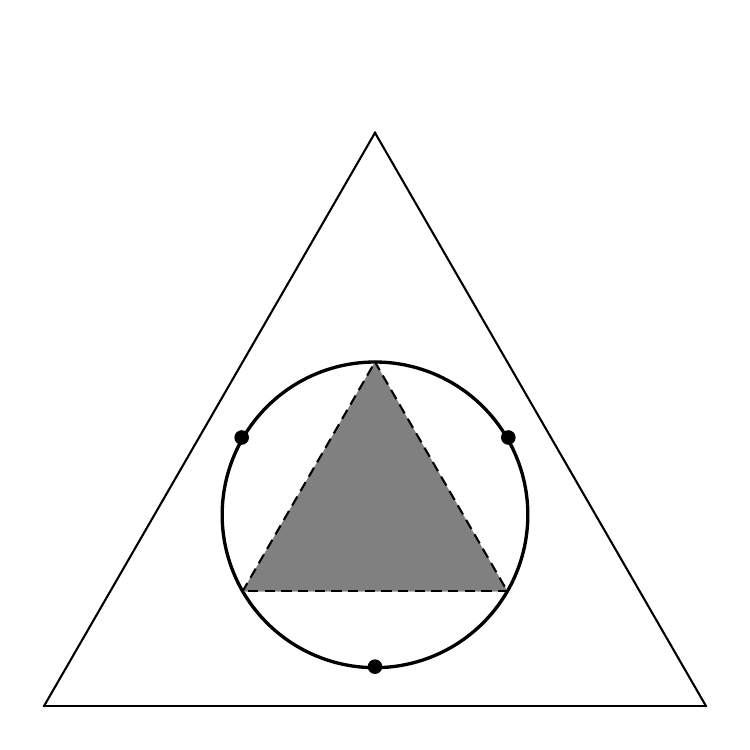}
\par\end{centering}
\centering{}\includegraphics[scale=0.5]{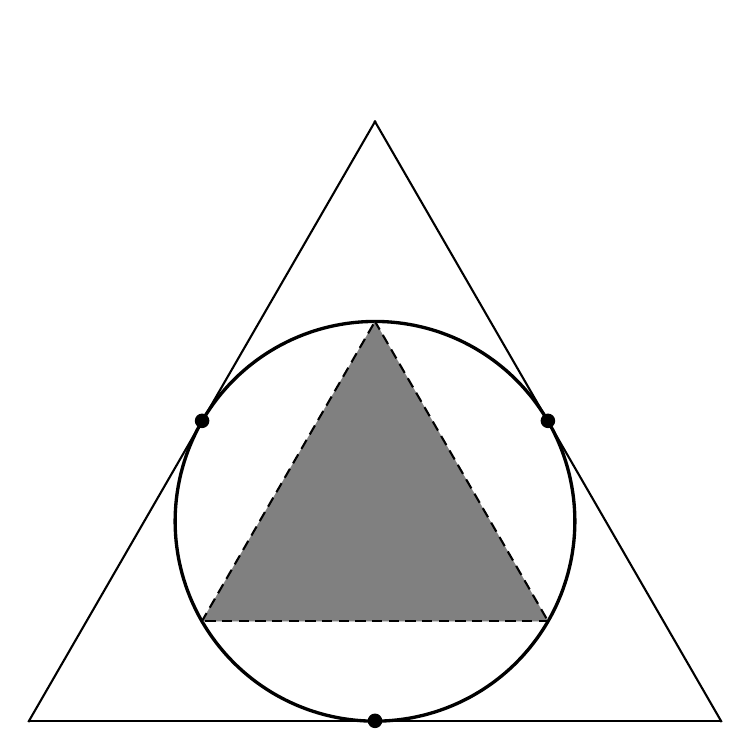}\includegraphics[scale=0.5]{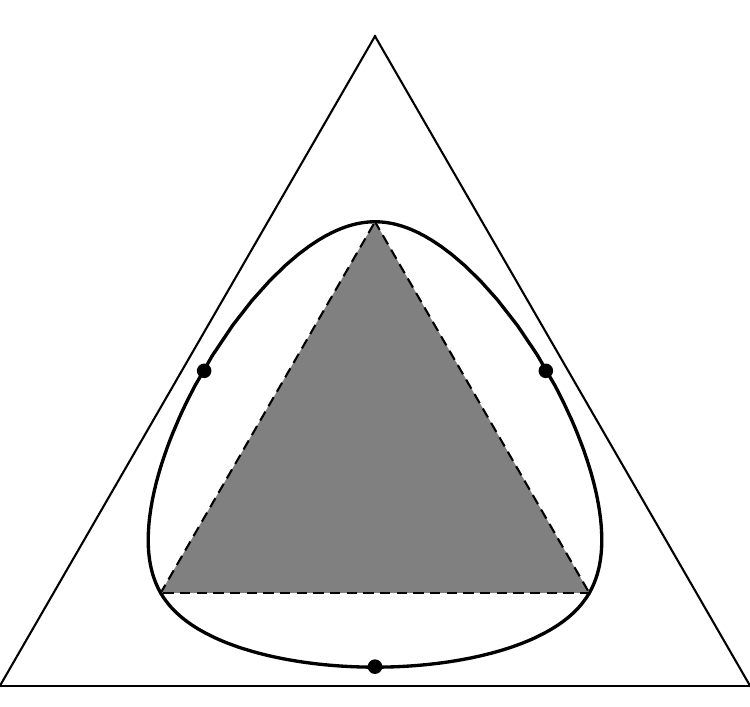}\caption{\label{fig:fishers} These figures show boundaries of $\tilde{\protect\F}_{\theta_{0}}^{(n)}:=\left\{ (F_{11},F_{22},F_{33})\mid\,F\in J^{(U)^{-1/2}}\protect\F_{\theta_{0}}^{(n)}J^{(U)^{-1/2}},\,\protect\Tr F=n^{2}+2n\right\} $,
where $\protect\F_{\theta_{0}}^{(n)}$ is the set of classical Fisher
information matrices for three-dimensional parametric unitary channel
models $\left\{ \Gamma_{\theta}^{\otimes n}:\protect\B(\protect\C^{2})\to\protect\B(\protect\C^{2})\mid\,\theta\in\Theta\subset\protect\R^{3}\right\} $
for $n=2$ (top-left), $n=3$ (top-right), $n=4$ (bottom-left), and
$n=5$ (bottom-right). The inner triangles painted in gray indicate
classical Fisher information matrices that can be realized by using
the randomized channel measurements given in Theorem \ref{thm:optM3n}
for $n=3,4,5$. For $n=2$, the inner triangle is same as $\tilde{\protect\F}_{\theta_{0}}^{(n)}$,
and it is realized by using randomized channel measurements given
in Theorem \ref{thm:n2}. The curves surrounding the inner triangles
are the boundaries of $\tilde{\protect\F}_{\theta_{0}}^{(n)}$ for
$n=3,4,5$. The three black dots correspond to optimal channel measurements
when $d=2$ and weight matrix $W=J_{\theta_{0}}^{(U)}$ given in Theorem
\ref{thm:gill_massar2}. }
\end{figure}

\section{Conclusion\label{sec:Conclusion}}

In this paper, we defined a Fisher information matrix $J_{\theta_{0}}^{(U)}$
at $\theta_{0}\in\Theta$ for the $SU(2)$ channel model $\G^{(n)}:=\left\{ \Gamma_{\theta}^{\otimes n}:\rho\mapsto U_{\theta}^{\otimes n}\rho U_{\theta}^{*\otimes n}\mid\theta\in\Theta\subset\R^{d}\right\} $
with $d=1,2,3$, and we showed a matrix inequality $n^{2}V_{\theta_{0}}[\psi,M,\hat{\theta}]\geq J_{\theta_{0}}^{(U)}$
for covariance matrix $V_{\theta_{0}}[\psi,M,\hat{\theta}]$ of locally
unbiased estimator $\left(\psi,M,\hat{\theta}\right)$. When $n=1$,
this equality is achievable by a maximally entangled input state.
When $n=2$ and $d=2$, this equality is achievable by a pure state
input given in Theorem \ref{thm:bound_matrix} without an ancilla
Hilbert space. 

When $d=3$, we proved a Gill and Massar type inequality $(n^{2}+2n)\Tr WV_{\theta_{0}}^{(n)}[\rho,M,\hat{\theta}]\geq c_{\theta_{0},W}$
with a lower bound $c_{\theta_{0},W}:=\left(\Tr\sqrt{J_{\theta_{0}}^{(U)^{-1/2}}WJ_{\theta_{0}}^{(U)^{-1/2}}}\right)^{2}$
for a $d\times d$ real positive matrix $W$. When $n\geq\max\left\{ 3,\frac{\sqrt{w_{2}}+\sqrt{w_{3}}}{\sqrt{w_{1}}}-1\right\} $,
the lower bound $c_{\theta_{0},W}$ is achievable by using a randomized
channel measurement without an ancilla Hilbert space, where $w_{1},w_{2},w_{3}$
are eigenvalues of $\tilde{W}:=J_{\theta_{0}}^{(U)^{-1/2}}WJ_{\theta_{0}}^{(U)^{-1/2}}$
such that $0\leq w_{1}\leq w_{2}\leq w_{3}$. When $d=3$ and $n=2$,
more informative and achievable bounds were obtained in Theorem \ref{thm:gill_massar_both}. 

When $d=2$ we proved $2\left\lfloor \frac{n^{2}+2n}{2}\right\rfloor \Tr WV_{\theta_{0}}^{(n)}[\rho,M,\hat{\theta}]\geq c_{\theta_{0},W}$,
and the equality is achievable if $W=J_{\theta_{0}}^{(U)}$. For an
arbitrary weight $W$ and $d=2$, the lower bound $c_{\theta_{0},W}$
is asymptotically achievable by using a sequence of randomized channel
measurements. 

There are many examples of $SU(2)$ structures in nature. For example,
a three dimensional magnetic field has $SU(2)$ structures, and its
estimation has been well studied\cite{rev9,rev12,rev13}. The results
of this paper could update these research in that it addresses general
arbitrary weights and general strategies. This paper is limited to
the $SU(2)$ channel model among channel estimation. However, as can
be seen from the Fig. \ref{fig:fishers}, the $SU(2)$ model estimation
with arbitrary weights has a sufficiently complex structure. To extend
the theory of this paper to the $SU(d)$ channel model with $d\geq3$
seems to be a very challenging problem. 

\appendix

\section*{Appendices}

\section{Generalized purification \label{sec:purification}}

In this appendix, we show a generalized purification as follows.

\begin{lem}
\label{lem:to_pure}For any quantum state $\rho\in\S(\H\otimes\H_{a})$
on a tensor product Hilbert space $\H\otimes\H_{a}$ of two Hilbert
spaces $\H$ and $\H_{a}$, there exists another Hilbert space $\H_{b}$
and a pure state $\ket{\psi}\bra{\psi}\in\S(\H\otimes\H_{b})$ on
$\H\otimes\H_{b}$ and a quantum channel $\Lambda:\B(\H_{b})\to\B(\H_{a})$
such that
\begin{equation}
\id_{\H}\otimes\Lambda(\ket{\psi}\bra{\psi})=\rho\label{eq:cest_lem_pure}
\end{equation}
and
\begin{equation}
\dim\H_{b}\leq\dim\H.\label{eq:cest_lem_dim}
\end{equation}
\end{lem}

\begin{proof}
Let 
\[
\rho=\sum_{i=1}^{r}p_{i}\ket{e_{i}}\bra{e_{i}}
\]
be a spectral decomposition of $\rho$ with an orthonormal basis $\left\{ \ket{e_{i}}\right\} _{i}$
of $\H\otimes\H_{a}$ and non-negative values $\left\{ p_{i}\right\} _{i}$
such that $\sum_{i}p_{i}=1$. Let us consider a vector
\[
\ket{\psi'}=\sum_{i=1}^{r}\sqrt{p_{i}}\ket{e_{i}}\ket{e_{i}'}\in\H\otimes\H_{a}\otimes\H_{a}'
\]
with an orthonormal basis $\left\{ \ket{e_{i}'}\right\} _{i=1}^{r}\subset\H_{a}'$
of a Hilbert space $\H_{a}'$ ($\dim\H_{a}'=r$). Note that ${\rm Tr}_{\H_{a}'}\ket{\psi'}\bra{\psi'}=\rho$.
By using the singular value decomposition, $\ket{\psi'}$ can be decomposed
to
\[
\ket{\psi'}=\sum_{k=1}^{s}\sqrt{q_{k}}\ket{f_{k}}\ket{f_{k}'},
\]
where $\left\{ \ket{f_{k}}\right\} _{k=1}^{s}\subset\H$ and $\left\{ \ket{f_{k}'}\right\} _{k=1}^{s}\subset\H_{a}\otimes\H_{a}'$
are orthonormal sets, and $\left\{ q_{k}\right\} _{k}$ are positive
values such that $\sum_{k=1}^{s}q_{k}=1$ with $s\leq\dim\H$. 

Let $\H_{b}$ be an $s$-dimensional Hilbert space with an orthonormal
basis $\left\{ \ket{g_{k}}\right\} _{k=1}^{s}\subset\H_{b}$, and
let $E:\H_{b}\rightarrow\H_{a}\otimes\H_{a}'$ be an embedding such
that $E\ket{g_{k}}=\ket{f_{k}'}$ for $k=1,\dots,s$. It can be seen
that a vector
\[
\ket{\psi}=\sum_{k=1}^{s}\sqrt{q_{k}}\ket{f_{k}}\ket{g_{k}}\in\H\otimes\H_{b}
\]
and a channel $\Lambda$ defined by
\[
\Lambda(B)={\rm Tr}_{\H_{a}'}EBE^{*}\qquad(B\in\B(\H_{b}))
\]
satisfy (\ref{eq:cest_lem_pure}). In fact,
\begin{align*}
\id_{\H}\otimes\Lambda(\ket{\psi}\bra{\psi}) & ={\rm Tr}_{\H_{a}'}(I\otimes E)\ket{\psi}\bra{\psi}(I\otimes E)^{*}\\
 & ={\rm Tr}_{\H_{a}'}\ket{\psi'}\bra{\psi'}=\rho.
\end{align*}
\end{proof}

\section{SLD bound and Holevo bound\label{sec:SLD_Holevo}}

Let $\left\{ \rho_{\theta}\in\S(\H)\mid\,\theta\in\Theta\subset\R^{d}\right\} $
be a quantum statistical model on a finite dimensional Hilbert space
$\H$. It is known that any locally unbiased estimator $(M,\hat{\theta})$
at $\theta_{0}\in\Theta$ satisfies inequalities
\begin{align}
\Tr WV_{\theta_{0}}[M,\hat{\theta}] & \geq C_{\theta_{0},W}^{(H)}\label{eq:holevo_bound}\\
 & \geq C_{\theta_{0},W}^{(S)},\label{eq:sld_bound}
\end{align}
where $C_{\theta_{0},W}^{(H)}$ is the Holevo bound for a weight $W>0$
defined by
\begin{align}
C_{\theta_{0},W}^{(H)} & :=\min_{B}\left\{ \left.\Tr WZ(B)+\Tr\left|\sqrt{W}{\rm Im}Z(B)\sqrt{W}\right|\right|Z_{ij}(B)=\Tr\rho_{\theta_{0}}B_{j}B_{i},\right.\\
 & \qquad\left.B_{1},\dots,B_{d}\text{ are Hermitian operators on \ensuremath{\H} such that }\Tr\partial_{i}\rho_{\theta_{0}}B_{j}=\delta_{ij}\right\} ,\label{eq:holevo_def}
\end{align}
and $C_{\theta_{0},W}^{(S)}$ is the SLD bound defined by $\Tr WJ_{\theta_{0}}^{(S)^{-1}}$
with the SLD Fisher information matrix $J_{\theta_{0}}^{(S)}$. In
general, the equation of (\ref{eq:holevo_bound}) is not achievable,
and Nagaoka-Hayashi bound and conic programming bound have been proposed
as more improved lower bounds \cite{nagaoka_hayashi,conic}. For pure
state models treated in this paper, the equation of (\ref{eq:holevo_bound})
is achievable (see Appendix \ref{sec:pure_holevo}). The following
lemma is known as a relation between the Holevo bound and the SLD
bound. 
\begin{lem}
\label{lem:sld_holevo}The Holevo bound and the SLD bound of any quantum
statistical model $\left\{ \rho_{\theta}\in\S(\H)\mid\,\theta\in\Theta\subset\R^{d}\right\} $
satisfy 
\[
C_{\theta_{0},W}^{(H)}=C_{\theta_{0},W}^{(S)}
\]
 for any weight $W$ if and only if
\[
\Tr\rho_{\theta_{0}}L_{\theta_{0},i}L_{\theta_{0},j}\in\R
\]
for $1\leq i,j\leq d$, where $L_{\theta_{0},i}$ is the SLD at $\theta_{0}\in\Theta$
in the $i$th direction. 
\end{lem}

\begin{proof}
We can see that observable $B_{i}$ in (\ref{eq:holevo_def}) satisfies
\[
\Tr\partial_{i}\rho_{\theta_{0}}B_{j}={\rm Re}\Tr\rho_{\theta_{0}}L_{\theta_{0},i}B_{j}=\delta_{ij}.
\]
It can be seen that
\begin{align}
C_{\theta_{0},W}^{(S)} & =\min_{B}\left\{ \left.\Tr WZ(B)\right|Z_{ij}(B)=\Tr\rho_{\theta_{0}}B_{j}B_{i},\right.\\
 & \qquad\left.B_{1},\dots,B_{d}\text{ are Hermitian operators on \ensuremath{\H} such that }\Tr\partial_{i}\rho_{\theta_{0}}B_{j}=\delta_{ij}\right\} ,\label{eq:sld_min}
\end{align}
and the minimum is achieved by observables defined by
\[
B_{i}^{(S)}:=\sum_{j=1}^{d}\left[J_{\theta}^{(S)^{-1}}\right]^{ij}L_{\theta_{0},j}
\]
because
\begin{align*}
{\rm Re}Z(B)_{ij} & ={\rm Re}\Tr\rho_{\theta_{0}}\left(B_{j}^{(S)}+K_{j}\right)\left(B_{i}^{(S)}+K_{i}\right)\\
 & ={\rm Re}\Tr\rho_{\theta_{0}}B_{j}^{(S)}B_{i}^{(S)}+{\rm Re}\Tr\rho_{\theta_{0}}K_{j}K_{i}\\
 & ={\rm Re}Z(B^{(S)})_{ij}+{\rm Re}Z(K)_{ij}\\
 & ={\rm Re}\left[J_{\theta}^{(S)^{-1}}\right]_{ij}+{\rm Re}Z(K)_{ij},
\end{align*}
where $K_{i}=B_{i}-B_{i}^{(S)}$. Therefore $C_{\theta_{0},W}^{(S)}=C_{\theta_{0},W}^{(H)}$
if and only if $\Tr\left|\sqrt{W}{\rm Im}Z(B^{(S)})\sqrt{W}\right|=0$.
This condition is equivalent to $\Tr\rho_{\theta_{0}}L_{\theta_{0},i}L_{\theta_{0},j}\in\R$. 
\end{proof}

\section{Holevo bound for pure state model\label{sec:pure_holevo}}
\begin{thm}
\label{thm:pure_holevo}Let $\left\{ \rho_{\theta}\mid\,\theta\in\Theta\subset\R^{d}\right\} $
be a quantum statistical model comprising pure states on a finite
dimensional Hilbert space $\H$, and let $C_{\theta_{0},W}^{(H)}$
be the Holevo bound at $\theta_{0}\in\Theta$ for a given weight $W>0$.
There exist a locally unbiased estimator $\left(M,\hat{\theta}\right)$
at $\theta_{0}\in\Theta$ such that
\begin{align}
\Tr WV_{\theta_{0}}[M,\hat{\theta}] & =\Tr WJ_{\theta_{0}}^{(C,M)^{-1}}=C_{\theta_{0},W}^{(H)}\\
 & \geq C_{\theta_{0},W}^{(S)},\label{eq:pure_holevo2}
\end{align}
where $J_{\theta_{0}}^{(C,M)}$ is the classical Fisher information
matrix with restrict to the POVM $M$. 

The inequality (\ref{eq:pure_holevo2}) for any $W>0$ implies
\begin{equation}
J_{\theta_{0}}^{(C)}(M)\leq J_{\theta_{0}}^{(S)},\label{eq:pure_sld_bound}
\end{equation}
and the equality is achievable by a projection valued measurement
$M$ if and only if SLDs $\left\{ L_{\theta_{0},i}\right\} _{i=1}^{d}$
satisfy
\begin{equation}
\Tr\rho_{\theta_{0}}L_{\theta_{0},i}L_{\theta_{0},j}\in\R\label{eq:pure_sld_achieve}
\end{equation}
for $1\leq i,j\leq d$.
\end{thm}

\begin{proof}
Let $B_{1},\dots,B_{d}$ be observables that attain the minimum in
(\ref{eq:holevo_def}), and let $\ket{\xi_{i}}:=B_{i}\ket{\psi}\in\H$
be vectors for $1\leq i\leq d$ with $\ket{\psi}\in\H$ such that
$\rho_{\theta_{0}}:=\ket{\psi}\bra{\psi}$. Let
\[
\hat{Z}:=\sqrt{W}^{-1}\left|\sqrt{W}{\rm Im}Z(B)\sqrt{W}\right|\sqrt{W}^{-1}-\ii{\rm Im}Z(B)
\]
be a positive-semidefinite matrix, and let $\hat{r}={\rm rank}\hat{Z}$.
There exits a Hilbert space $\hat{\H}=\C^{\hat{r}+1}$ and vectors
$\ket{\hat{\psi}},\ket{\hat{\xi}_{1}},\dots,\ket{\hat{\xi}_{d}}\in\H_{a}$
such that
\begin{align*}
\braket{\hat{\xi}_{i}}{\hat{\xi}_{j}} & =\hat{Z}_{ij}\qquad(1\leq i,j\leq d),\\
\braket{\hat{\psi}}{\hat{\xi}_{i}} & =0\qquad(1\leq i\leq d),\\
\braket{\hat{\psi}}{\hat{\psi}} & =1.
\end{align*}
Let 
\begin{align*}
\ket{\tilde{\xi}_{i}} & :=\ket{\xi_{i}}\otimes\ket{\hat{\psi}}+\ket{\psi}\otimes\ket{\hat{\xi}_{i}}\qquad(1\leq i\leq d),\\
\ket{\tilde{\psi}} & :=\ket{\psi}\otimes\ket{\hat{\psi}}
\end{align*}
be vectors in $\H\otimes\hat{\H}$. Because $\braket{\tilde{\psi}}{\tilde{\xi}_{i}}(=0)$
and $\braket{\tilde{\xi}_{i}}{\tilde{\xi}_{j}}(={\rm Re}Z(B)_{ji}+{\rm Re}\hat{Z}_{ji})$
are all real, there exist an orthonormal basis $\left\{ \ket{\tilde{e}_{k}}\right\} _{k=1}^{\dim\H\otimes\hat{\H}}$
of $\H\otimes\hat{\H}$ such that $\braket{\tilde{e}_{k}}{\tilde{\psi}}$
and $\braket{\tilde{e}_{k}}{\tilde{\xi}_{i}}$ are all real, and that
$\braket{\tilde{e}_{k}}{\tilde{\psi}}\not=0$ for all $k$. Let
\[
\tilde{T}_{i}:=\theta_{0}^{i}I_{\H\otimes\hat{\H}}+\sum_{k=1}^{\dim\H\otimes\hat{\H}}\frac{\braket{\tilde{e}_{k}}{\tilde{\xi}_{i}}}{\braket{\tilde{e}_{k}}{\tilde{\psi}}}\ket{\tilde{e}_{k}}\bra{\tilde{e}_{k}}\qquad(1\leq i\leq d),
\]
be observables on $\H\otimes\hat{\H}$. Then $\tilde{T}_{1},\dots,\tilde{T}_{d}$
are simultaneously measurable, and satisfy the local unbiasedness
condition:
\[
\Tr\left(\rho_{\theta_{0}}\otimes\ket{\hat{\psi}}\bra{\hat{\psi}}\right)\tilde{T}_{i}=\bra{\tilde{\psi}}\tilde{T}_{i}\ket{\tilde{\psi}}=\theta_{0}^{i}
\]
and
\[
\Tr\left(\partial_{i}\rho_{\theta_{0}}\otimes\ket{\hat{\psi}}\bra{\hat{\psi}}\right)\tilde{T}_{j}={\rm Re}\bra{\tilde{l}_{i}}\tilde{T}_{j}\ket{\tilde{\psi}}=\delta_{ij},
\]
where 
\[
\ket{\tilde{l}_{i}}=\left(L_{\theta_{0},i}\ket{\psi}\right)\otimes\ket{\hat{\psi}}.
\]
Further they satisfy
\begin{align*}
V_{\theta_{0}}[\tilde{T}]_{ij} & =\Tr\left(\rho_{\theta_{0}}\otimes\ket{\hat{\psi}}\bra{\hat{\psi}}\right)\left(\tilde{T}_{i}-\theta_{0}^{i}I_{\H\otimes\hat{\H}}\right)\left(\tilde{T}_{j}-\theta_{0}^{j}I_{\H\otimes\hat{\H}}\right)\\
 & =\braket{\tilde{\xi}_{i}}{\tilde{\xi}_{j}}={\rm Re}Z(B)_{ji}+{\rm Re}\hat{Z}_{ji}\\
 & =\left[{\rm Re}Z(B)+\sqrt{W}^{-1}\left|\sqrt{W}{\rm Im}Z(B)\sqrt{W}\right|\sqrt{W}^{-1}\right]_{ji}.
\end{align*}
This implies $\Tr WV_{\theta_{0}}[\tilde{T}]=C_{\theta_{0},W}^{(H)}$. 

The proof of the equality achievement condition for (\ref{eq:pure_sld_bound})
is given in Lemma \ref{lem:sld_holevo}. Note that when $\Tr\rho_{\theta_{0}}L_{\theta_{0},i}L_{\theta_{0},j}\in\R$
for $1\leq i,j\leq d$, ${\rm Im}Z(B)=0$ and $\dim\hat{\H}=1$. Thus,
$\left\{ \tilde{T}_{i}\right\} _{i=1}^{d}$ are observables on $\H$,
and they correspond to a projection valued measurement on $\H$. 
\end{proof}

\section{Numerical calculations for the curves in Figure \ref{fig:fishers}
\label{sec:numerical_calculations}}

In this appendix, we details about the numerical calculations to obtain
the curves in Figure \ref{fig:fishers}. Without loss of generality,
we can assume $X_{1},X_{2},X_{3}$ are Pauli matrices and $J_{\theta_{0}}^{(U)}=I$.
To reveal the boundary of the convex set $\tilde{\F}_{\theta_{0}}^{(n)}$,
we calculate 
\begin{equation}
F^{(i)}(t):=\argmax_{F\in\F_{\theta_{0}}^{(n)}}\Tr H_{t}^{(i)}F\label{eq:numaricalF}
\end{equation}
numerically for $t\in[0,1]$ and $i=1,2,3$, where $H_{t}^{(1)}={\rm diag}(0,t,1-t)$,
$H_{t}^{(2)}={\rm diag}(1-t,0,t)$, $H_{t}^{(3)}={\rm diag}(t,1-t,0)$.
We can see that $F^{(i)}(t)$ corresponds to a point on the boundary
of $\tilde{\F}_{\theta_{0}}^{(n)}$. To calculate $F^{(3)}(t)$, we
compute an unit eigenvector $\ket t$ of $(tX_{1}^{(n)^{2}}+(1-t)X_{2}^{(n)^{2}})\otimes I$
with respect to the maximum eigenvalue with a $2\times2$ identity
matrix $I$. Finally, we check that $\bra tX_{i}^{(n)}X_{j}^{(n)}\otimes I\ket t$
is a real value and $\bra tX_{i}^{(n)}\otimes I\ket t=0$ for $1\leq i,j\leq3$
to verify $\ket t\in\A_{\theta_{0}}^{(n)}$, then we obtain $F^{(3)}(t)=\left[\bra tX_{i}^{(n)}X_{j}^{(n)}\otimes I\ket t\right]_{ij}$.
The calculation of $F^{(1)}(t)$ and $F^{(2)}(t)$ are similar. 

\section*{Acknowledgments}

The present study was supported by JSPS KAKENHI Grant Numbers JP23H01090,
JP22K03466.

\end{document}